\documentclass[12pt,oneside,english,12pt, oneside, reqno]{amsart}
\usepackage[T1]{fontenc}
\usepackage[latin9]{inputenc}
\usepackage{xcolor}
\usepackage{float}
\usepackage{amstext}
\usepackage{amsthm}
\usepackage{amssymb}
\usepackage{graphicx}
\usepackage{geometry}
\usepackage{setspace}
\usepackage[authoryear,round,comma]{natbib}
\PassOptionsToPackage{normalem}{ulem}
\usepackage{ulem}
\makeatletter

\floatstyle{ruled}
\newfloat{algorithm}{tbp}{loa}
\providecommand{\algorithmname}{Algorithm}
\floatname{algorithm}{\protect\algorithmname}
\providecolor{lyxadded}{rgb}{0,0,1}
\providecolor{lyxdeleted}{rgb}{1,0,0}
\DeclareRobustCommand{\mklyxadded}[1]{\textcolor{lyxadded}\bgroup#1\egroup}
\DeclareRobustCommand{\mklyxdeleted}[1]{\textcolor{lyxdeleted}\bgroup\mklyxsout{#1}\egroup}
\DeclareRobustCommand{\mklyxsout}[1]{\ifx\\#1\else\sout{#1}\fi}
\theoremstyle{definition}
\newtheorem{defn}{\protect\definitionname}
\theoremstyle{plain}
\newtheorem{prop}{\protect\propositionname}
\theoremstyle{plain}
\newtheorem*{assumption*}{\protect\assumptionname}

\usepackage{lmodern}
\usepackage[T1]{fontenc}
\usepackage{braket}
\usepackage{hyperref}
\hypersetup{
colorlinks=false,
citecolor=green
}
\usepackage{algorithm}
\usepackage[noend]{algpseudocode}

\makeatother

\usepackage{babel}
\providecommand{\assumptionname}{Assumption}
\providecommand{\definitionname}{Definition}
\providecommand{\propositionname}{Proposition}

\begin{document}
\title[Limited belief propagation]{Limited belief propagation \\and contingent thinking}
\date{7 June 2026}
\author{Andrew Ellis and Ran Spiegler}
\begin{abstract}
An agent updates her beliefs over a set of variables after observing
some of them. We provide a representation of updated beliefs that
captures limited propagation of her observation's implications through
the directed acyclic graph that represents the relations between all
variables. Failure of contingent thinking occurs when she performs
fewer inference steps from unobserved variables than observed ones,
leading to correlation neglect and violations of iterated expectations.
Our framework offers a new perspective on existing experiments about
contingent thinking and suggests new directions. We characterize the
model's relationship with familiar Bayesian and non-Bayesian benchmarks,
and illustrate it with applications to public-good provision and social
learning games.\medskip{}
\end{abstract}

\thanks{Ellis: LSE, a.ellis@lse.ac.uk. Spiegler: Tel Aviv University and UCL,
rani@tauex.tau.ac.il. Spiegler acknowledges financial support from
UKRI Frontier Research Grant no. EP/Y033361/1. We thank Yair Weiss,
as well as seminar and conference audiences in Budapest and Vienna,
for helpful comments.\bigskip{}
}
\maketitle

\section{Introduction}

People make systematic belief errors in settings where correct inference
requires conditioning on hypothetical contingencies. For example,
bidders in auctions suffer from the ``winner's curse'': they fail
to account for the information about the other bidders' types given
by the hypothetical event of winning the auction. In this paper, we
propose to model these failures, based on the postulate that people
do not think through all the implications of their information and
choices. 

To illustrate our approach, consider a second-price auction with common
values, of the kind taught in basic Game Theory courses. \citet[pp 297-9]{osborne2004introduction}
describes a game with two bidders, $1$ and $2$, that each observes
an independent private signal $t_{i}$ before choosing a bid $b_{i}$.
The bids and the auction rules determine the object's allocation $z$.
The object's common value is $t_{1}+t_{2}$. Bidders' strategies induce
a joint probability distribution over the five variables, which individual
bidders rely on for their strategic calculations. Specifically, to
bid rationally, bidder $1$ must calculate the (joint) probability
of winning, her payment, and player 2's type given her type and each
possible bid. Osborne fixes bidder $1$'s bid $b_{1}$ and type $t_{1}$
and computes it as follows. First, calculate bidder $1$'s probability
of winning the auction given $b_{1}$; second, find the distribution
over $b_{2}$ conditional on bidder $1$ winning the auction with
the bid $b_{1}$; and third, derive the conditional distribution of
$t_{2}$ given $b_{2}$ using bidder $2$'s strategy. Combining these
calculations, we get
\begin{equation}
p(t_{2},b_{2},z\mid t_{1},b_{1})=p(z\mid b_{1})p(b_{2}\mid b_{1},z)p(t_{2}\mid b_{2}),\label{eq:bayes intro}
\end{equation}
which can be used to verify that a conjectured strategy is optimal.

This sequence of calculations is an example of a procedure called
\textit{belief propagation}. It simplifies the process of Bayesian
updating using a \emph{directed acyclic graph (DAG),} in this case
\begin{equation}
t_{1}\rightarrow b_{1}\rightarrow z\leftarrow b_{2}\leftarrow t_{2},\label{eq: dag intro}
\end{equation}
that represents the underlying causal relationships between variables,
following the literature on graphical probabilistic models, e.g.,
\citet{cowell1999probabilistic}, \citet{pearl2009causality}, or
\citet{koller2009probabilistic}. The procedure breaks a large computational
task into more manageable, smaller pieces. It has three salient features.
The variables are updated one at a time in sequence ($z$, then $b_{2}$,
then $t_{2}$). Each calculation is local, in the sense that it involves
variables that are near each other according to the DAG (no step uses
$t_{1}$, and the last step leaves out $z$ and $b_{1}$ as well).
Each step treats the variables that have already been updated in the
same way, regardless of whether they are known ($t_{1},b_{1}$) or
inferred ($z,b_{2}$). 

We study probabilistic inference by a decision-maker (DM) who tries
to perform this belief propagation procedure but who is limited in
the ability to execute it. She still considers variables sequentially
but may consider only some of the graph paths to each variable. This
means that at each step, fewer variables are nearby and each affects
others in fewer ways. For example, rather than propagating information
through all paths in the DAG, the DM may only propagate information
from one or two steps away. 

Suppose that in the auction example, bidder $1$ can perform only
two steps of inference from her choice variable $b_{1}$ and her type
variable $t_{1}$, but cannot perform any steps from the hypothetical
variables $z$ and $b_{2}$. This captures two forms of bounded rationality:
inability to perform long chains of reasoning, and inability to draw
inferences from hypothetical information. The inference proceeds by
considering $z$, then $b_{2}$, and finally $t_{2}$. Since $z$
is sufficiently close to both $t_{1}$ and $b_{1}$, the DM updates
it to $p(z\mid b_{1},t_{1})$ (which equals $p(z\mid b_{1})$ because
$t_{1}$ is independent of $z$ given $b_{1}$). Turning to $b_{2}$,
this variable has a short enough path from $b_{1}$ ($b_{1}\rightarrow z\leftarrow b_{2}$).
However, this path treats $b_{1}$ and $b_{2}$ as mutually independent;
and moreover, the DM does not draw any inference about $b_{2}$ from
the hypothetical variable $z$. Therefore, the DM updates $b_{2}$
to $p(b_{2})$. Finally, $t_{2}$ is is not close enough to either
$t_{1}$ or $b_{1}$, and the DM does not draw any inference about
$t_{2}$ from the hypothetical variables $z$ and $b_{2}$. Therefore,
she updates $t_{2}$ to $p(t_{2})$.

Combining these calculations, bidder $1$'s conditional belief will
thus be
\begin{equation}
\hat{p}\left(t_{2},b_{2},z\mid t_{1},b_{1}\right)=p\left(z\mid b_{1},t_{1}\right)p\left(b_{2}\right)p\left(t_{2}\right)\label{eq: lpb auction}
\end{equation}
This inference is incorrect: winning the auction with a low $b_{1}$
implies low $b_{2}$ and $t_{2}$, but bidder $1$ fails to draw \textit{any}
inference about either. These belief errors result in the winner's
curse: the DM fails to draw correct inferences about the opponent's
bid and type conditional on winning the auction. 

We propose a general procedure for updating a distribution over a
collection of variables consistent with a DAG. Some of these variables
are ``givens'' (signals or actions) known to the DM, while the others
are hypothetical. The DM calculates a conditional belief over the
hypothetical variables as a function of the values of the givens.
\emph{Limited-propagation beliefs} consider the hypothetical variables
in sequence, according to a fixed ordering $\succ$. For every pair
of nodes $i$ and $j$ in the DAG, there is a subset of the paths
in the DAG from $i$ to $j$ (paths can involve links in either direction),
referred to as the \emph{flows} from $i$ to $j$. The flows capture
the relationships considered when passing information to $j$. The
sparser the flows to $j$, the more limited the DM's ability to draw
inferences about it. For instance, in the auction example above, the
flows from $b_{1}$ to $b_{2}$ and $z$ consist of the paths $b_{1}\rightarrow z$
and $b_{1}\rightarrow z\leftarrow b_{2}$; there is no flow from $b_{1}$
to $t_{2}$; and there are no flows from the hypothetical variables
$z$ and $b_{1}$. 

The DM's conditional belief over the variable $j$ is formed as if
the prior distribution is consistent with a sub-DAG consisting of
all the flows into $j$, conditioned on all the $\succ$-earlier variables
that flow to $j$. Sometimes, as in the above auction example, the
sub-DAGs associated with different nodes are mutually consistent and
can be subsumed in a single (potentially misspecified) DAG. In other
cases, limited-propagation beliefs \textit{cannot} be rewritten as
if the DM reasons with a single subjective DAG. We provide a procedural
foundation for the representation, a slight modification of the Variable
Elimination algorithm that can be found in many texts on Bayesian
networks, e.g., \citet{koller2009probabilistic}.

We use the specification of flows to capture intuitions about what
makes a sequence of local inferences hard for DMs. While the correct
procedure can go in either causal direction, our DM may find it harder
to travel backwards than forwards. While the correct procedure propagates
beliefs from a node to all its neighbors, our DM may only be able
go to one of them. While the correct procedure propagates information
from any number of givens, our DM's ability to do so may be diminished
when they are numerous. While the correct procedure does not discriminate
between observed and unobserved variables, our DM may find it harder
to propagate beliefs from the latter.

The structure of the causal relationships between variables drives
departures from Bayesian updating in limited-propagation beliefs.
This distinguishes our approach from updating models in which departures
from Bayes' rule depend on the properties of the revealed event \citep[cf.][etc.]{ortoleva2012modeling,ortoleva2024alternatives,gennaioli2010what,zhao2022pseudo}.
It thus speaks to a distinct class of anomalies for Bayesianism, referred
to as \textit{failures of contingent thinking}. In a survey article,
\citet{niederle2023cognitive} define this phenomenon as occurring
when the DM ``does optimize in a presentation of the problem that
helps them focus on all relevant contingencies {[}such as providing
additional information{]}... but does not optimize if the problem
is presented without such aids.'' Limited-propagation beliefs exhibit
failures of contingent thinking when they propagate hypotheticals
less than givens.

We use this feature to suggest new commonalities in the experimental
findings. For example, the classic ``Acquire a Company'' experiments
\citep{samuelson1985winner} share a DAG with ``Monty Hall'' experiments
\citep{friedman1998monty}, and subjects fail to update correctly
in both. However, the two-company version of ``Acquire a Company''
\citep{martinez2019failures}, in which more subjects act like Bayesians,
has a different underlying DAG. Limited-propagation beliefs can explain
both failures to update correctly in the former instances and correct
updating in the latter. It provides a similarly unifying perspective
into experimental evidence on pivotal reasoning that includes \citet{kagel1987information,esponda2014hypothetical,li2017obviously,ngangoue2021learning,calford2024contingent}.
Our framework not only organizes existing experimental findings, but
also suggests novel directions for experimental work, as we illustrate
with an example based on the recreational puzzle Kakuro.\footnote{The idea that belief anomalies (including the Monty Hall fallacy)
can be unified by failures to reason about a DAG is suggested in Ch.
6 of \citet{pearl2018book}.}

We then turn to general properties of limited-propagation beliefs.
We first study connections to existing updating models. Most similar
are ones based on failures to update correctly about the relationship
between actions and types in strategic settings \citep{eyster2005cursed,jehiel2008revisiting,esponda2008behavioral}.
As argued by \citet{spiegler2016bayesian}, these models can be thought
of as fitting a single subjective DAG to the objective distribution.
As mentioned above, limited-propagation beliefs do not generally have
this property. However, we provide a sufficient condition for it to
hold. In particular, this condition holds whenever the non-directed
version of the true DAG is a tree (as in the above auction example).

Using this result and focusing on the case that the true process is
consistent with a causal chain, we demonstrate that limited-propagation
beliefs have a simple testable implication. Correct inference about
a variable that is represented by a distant node implies correct inference
about a variable that is represented by a nearer one. We also use
it to show that expanding the DM's flows (so that she can perform
longer chains of inference) or her set of givens (so that she becomes
better informed) need not increase belief accuracy. 

Finally, we investigate the effects of limited contingent thinking.
We show that contingent thinking is necessary for Bayesian updating
when there is more than one hypothetical variable. We provide simple
necessary and sufficient conditions for correct updating. Under the
assumption that the DM can perform \textit{no} contingent thinking
(captured by no flows from non-givens), we consider whether limited-propagation
beliefs satisfy \emph{iterated expectations}, a feature closely tied
to dynamic consistency and satisfied by Bayes' rule. The property
typically fails; we provide necessary and sufficient conditions for
it to hold. Finally, when the DM can perform no contingent thinking,
her conditional belief exhibits\emph{ correlation neglect.} Thus,
as with our discussion of experimental work, we see how the limited-propagation
framework offers a new perspective into belief errors and clarifies
their connection to limited contingent thinking. 

\section{The Model}

We structure this section as follows. \textit{Directed acyclic graphs}
(DAGs) are first introduced, and related definitions are provided.
Then, we present our ``limited-propagation beliefs'' representation.
Finally, examples illustrate some of its features.

\subsection{Preliminaries: DAGs}

Consider a random vector $X=(X_{1},...,X_{n})$ and an index set $N=\{1,...,n\}$.
A \textit{causal model} is represented by a DAG $(N,Q)$, where $N$
is a set of \textit{nodes} and $Q\subseteq N\times N$ is a set of
directed \textit{edges} that has no cycles. A node $i$ represents
the random variable $X_{i}$, and the edge $(i,j)$ means that $X_{i}$
is perceived as a direct cause of $X_{j}$. Throughout, we identify
a DAG with its set of edges and suppress its set of nodes. For $x=(x_{1},...,x_{n})$
and $E\subset N$, denote $x_{E}=(x_{i})_{i\in E}$. Abusing notation,
$Q(i)=\{j\in N\mid(j,i)\in Q\}$ is the set of \emph{parents} of node
$i$. 

For a probability distribution $p$ over $X$, the \emph{factorization}
of $p$ according to $Q$ is
\[
p_{Q}(x)=\prod\limits^{N}_{i=1}p\left(x_{i}\mid x_{Q(i)}\right).
\]
E.g., for $Q:$ $1\rightarrow3\rightarrow4\leftarrow2$, $p_{Q}(x)=p(x_{1})p(x_{2})p(x_{3}\mid x_{1})p(x_{4}\mid x_{2},x_{3})$.
We say that $p$ is\emph{ consistent} with $Q$ if $p=p_{Q}$, which
means that $p$ factorizes according to $Q$.

A \emph{path} from $j$ to $i$ in the DAG $Q$ is a sequence of nodes
$\left(i_{1},\dots,i_{K}\right)$ such that $i_{1}=j$, $i_{K}=i$,
and for every $k=1,...,K-1$, either $(i_{k},i_{k+1})$ or $(i_{k+1},i_{k})$
belongs to $Q$. We will often identify a path with the set of edges
it consists of. The \emph{skeleton} of a DAG $Q$, denoted $\tilde{Q}$,
is the undirected version of $Q$, i.e., $\tilde{Q}=\left\{ \{i,j\}\mid(i,j)\in Q\right\} $.
A \emph{v-collider} in $Q$ is a triple of nodes $(i,j,k)$ so that
$(i,j),(k,j)\in Q$ and $\{i,k\}\notin\tilde{Q}$.

\subsection{Limited-propagation beliefs}

We model a decision-maker (DM) who forms beliefs about $X$. The DM's
prior is a probability measure $p$ over $X$ that is consistent with
a DAG $R$. The DAG $R$ represents the \textit{true} underlying causal
model. Let $G\subset N$ represent the set of \emph{given }variables,
which the DM chooses or observes prior to making her choice. For simplicity
in notation and to avoid issues of updating on zero probability events,
we assume that each variable has finite support and that $p$ has
full support; in examples, we sometimes deviate from these assumptions
provided that they do not cause difficulties with the formulas. When
it does not cause confusion, we sometimes refer to the variable $X_{i}$
by its index $i$, and use $x_{i}$ to denote realizations of $X_{i}$.

A conventional DM would form the posterior belief $p(x_{N-G}|x_{G})$
in accordance with Bayes' rule. We model a DM who instead updates
her beliefs by propagating information from given variables along
some paths in $R$. The propagation is limited in the sense that it
may halt too early. We offer a representation of \textit{limited-propagation
beliefs} that is based on two parameters.

First, $\succ$ is a \textit{linear }\emph{order} on $N$, which defines
the sequence in which the DM considers the variables. We refer to
$\succ$ simply as the ``order''. When $j$ precedes $i$, i.e.,
$i\succ j$, the DM considers $j$ before $i$. We require that the
nodes in $G$ are considered first: if $g\in G$ and $i\notin G$,
then $i\succ g$. 

Second, for every pair of distinct nodes $i,j\in N$, the \textit{flow}
$F_{j\rightarrow i}\subset R$ from $j$ to $i$ is a (possibly empty)
subset of the paths in $R$ from $j$ to $i$. The flow from $j$
to $i$ describes the agent's ability to draw inferences about $X_{i}$
from $X_{j}$. Expanded flows capture better ability to carry out
inferences. Flows may depend on the specification of $G$, but we
suppress this notation until it is relevant in Section \ref{subsec: contingent thinking}.
Denote by $R_{i}=\cup_{j\neq i}F_{j\rightarrow i}$ the flows into
$i$, noting that it is a sub-DAG of $R$.\footnote{Note that we can regard $F_{i\rightarrow j}$ as a subset of $R$,
because as mentioned above, we identify it with the set of edges it
consists of.}

We assume that if there is a path from $j$ to $k$ in $F_{j\rightarrow i}$,
then this path also belongs to $F_{j\rightarrow k}$. If information
from $j$ flows to $i$ through $k$, then it also flows to $k$.
Importantly, we do \textit{not} require that a path from $k$ to $i$
in $F_{j\rightarrow i}$ must belong to $F_{k\rightarrow i}$. The
latter property would rule out imperfect contingent thinking, when
$j\in G$ and $k\in N-G$.

In a specification we use repeatedly, $F_{j\rightarrow i}=\emptyset$
for all $j\notin G$, and $F_{j\rightarrow i}$ consists of all paths
from $j$ to $i$ that have at most $K$ edges when $j\in G$. The
auction example in the Introduction obeys this specification with
$K=2$: $G$ consists of $t_{1}$ and $b_{1}$, and the flows are
$F_{t_{1}\rightarrow x}=t_{1}\rightarrow b_{1}\rightarrow z$, $F_{b_{1}\rightarrow b_{2}}=b_{1}\rightarrow z\leftarrow b_{2}$,
$F_{t_{1}\rightarrow b_{1}}=t_{1}\rightarrow b_{1}$, $F_{b_{1}\rightarrow z}=b_{1}\rightarrow z$,
and $F_{i\rightarrow j}=\emptyset$ for all $i\notin G$. This captures
absence of contingent thinking, in the sense that the DM draws no
inferences from hypothetical variables. 

Flows are flexible objects that can capture various sources of difficulty
in updating beliefs over multiple variables. As mentioned above, they
can capture limited contingent reasoning by drawing a distinction
between flows from givens and from hypotheticals. Flows can depend
on the size of $G$, capturing the intuition that it is harder to
draw inferences from a larger set of givens. Flows may also depend
on variable salience; e.g., the DM may perform more inference steps
from a variable that serves as an explicit argument in her payoff
function. Flows may distinguish between ``downstream'' paths (which
track the causal flow) and ``upstream'' paths (which go against
the causal direction), capturing the intuition that the former are
easier to perform (see \citet{ambuehl2026mental} for experimental
evidence). Flows may also take the entire set of paths into account;
e.g., the DM may struggle to pursue a path from one node to another
when there are additional ``backdoor'' paths between them.

For every $i\in N-G$, define the set 
\[
M(i)=\left\{ j\in N\mid i\succ j\ and\ F_{j\rightarrow i}\neq\emptyset\right\} .
\]
The nodes in $M(i)$ represent the variables that both precede $X_{i}$
and flow to it. The DM deems these variables relevant for updating
her belief about $X_{i}$. Thus, flows not only specify inferential
paths between variables, but also determine whether the DM retains
a variable's updated value when she proceeds to update beliefs about
subsequent variables.
\begin{defn}
\label{def:LP beliefs}For flows $\{F_{j\rightarrow i}\}_{i\neq j}$
and order $\succ$, the conditional distribution $\hat{p}$ is a \emph{limited-propagation
belief} if 
\begin{equation}
\hat{p}\left(x_{N-G}\mid x_{G}\right)=\prod\limits_{i\in N-G}p_{R_{i}}\left(x_{i}\mid x_{M(i)}\right)\label{eq: subjective beliefs}
\end{equation}
for every $x_{N}$.\smallskip{}

In contrast, the Bayesian update of $p$ can be written as
\[
p\left(x_{N-G}|x_{G}\right)=\prod_{i\in N-G}p_{R}\left(x_{i}\mid x_{\{j\in N\mid i\succ j\}}\right)
\]
using the chain rule of probability and the assumption that $p$ is
consistent with $R$. Thus, a limited-propagation belief departs from
correct Bayesian updating in two ways. First, it conditions $x_{i}$
on a subset $M(i)$ of $\{j\in N\mid i\succ j\}$, instead of the
whole set. Variables that do not flow to $i$ are absent from the
conditioning set. Second, limited-propagation beliefs compute the
update given conditioned variables using only the subset $R_{i}$
of the pathways in $R$. Paths not in the flow are neglected when
forming beliefs.

Limited-propagation beliefs capture the following mental process.
Consider the earliest node $i\in N-G$ according to the order $\succ$.
A Bayesian DM would compute $p(x_{i}\mid x_{G})$ by considering all
of the ways in which the particular realization of the given variables
could be relevant for inferring or predicting $X_{i}$. This inference
can be executed via a sequence of local computations along the pathways
in $R$ from $G$ nodes to $i$. She then turns to the next node according
to $\succ$, denoted $i+1$ for the sake of this explanation. She
repeats the procedure for $X_{i+1}$ and computes $p\left(x_{i+1}|x_{G},x_{i}\right)$,
by considering all the ways in which any of the givens, \textit{or}
$X_{i}$, could be relevant for inferring or predicting $X_{i+1}$.
She proceeds to examine the subsequent variables in the same manner.

By comparison, a DM with limited-propagation beliefs only transfers
information from $G$ to $i$ via the smaller set of paths $R_{i}$,
and using only the givens that flow to $i$, namely $M(i)$. She computes
$\hat{p}\left(x_{i}\mid x_{G}\right)$ via a sequence of local computations
along the paths in $R_{i}$ from $M(i)$ to $i$. She then turns to
$i+1$, and repeats the same procedure. Now, she propagates the information
from $M(i+1)$ (which may or may not include $i$) along the paths
in $R_{i+1}$. The DM continues in this fashion down the ordering
$\succ$ until she covers all variables in $N-G$. In Section \ref{sec:Procedure}
we provide an algorithmically explicit account of this inference process.
\end{defn}

\subsection{Examples}

We now present examples that illustrate the limited-propagation belief
representation. Throughout the sub-section, we adopt the convention
of explicitly stating only the longest flows from a variable and filling
in the rest using the betweenness property. 

\subsubsection{Auction}

Returning to the running example, the DM is player $1$ in the common
value auction. The true DAG and Bayesian updates are as described
by (\ref{eq: dag intro}) and (\ref{eq:bayes intro}), respectively.
The givens and the flows are as described earlier in this section.
For an ordering that satisfies $z\prec b_{2}\prec t_{2}$, limited-propagation
beliefs are given by (\ref{eq: lpb auction}), as explained in the
Introduction.

The DM's limited-propagation belief results in under-appreciation
of the relationship between the probability of winning, the bid of
the other player, and the other player's type. This reflects the experimentally
documented poor performance of subjects in common-value, sealed-bid,
second-price auctions \citep{kagel1986winner,kagel1987information}.
To capture bidders' reasoning about the related format of an English
(ascending-price) auction, we can add $b_{2}$ to the DM's set of
givens. In line with our previous assumption that the DM can perform
two steps of inference from given variables, let $F_{b_{2}\rightarrow b_{1}}=b_{2}\rightarrow z\leftarrow b_{1}$
and $F_{b_{2}\rightarrow t_{2}}=b_{2}\rightarrow z\leftarrow b_{1}$;
flows from the other variables are unchanged. Under this alternative
specification, the limited-propagation belief coincides with the correct
Bayesian posterior, for any order $\succ$. This is consistent with
the experimental finding that subjects' behavior is closer to the
rational benchmark in English auctions \citep{levin1996Revenue,li2017obviously}.
We will revisit this observation in Section \ref{sec:Limited-propagation-and}.

\subsubsection{Monty Hall\label{subsec:example monty}}

In the Monty Hall problem \citep{selvin1975monty}, the DM is the
contestant. There are three doors, one of which conceals a prize.
The DM picks a door, the host opens an unpicked door without the prize,
and then the DM is asked whether she would like to switch doors. To
formalize this, $\theta$ indicates the location of the prize, $h$
indicates the DM's initial guess, and $d$ is the door opened. The
true DAG is $\theta\rightarrow d\leftarrow h,$ which is a $v$-collider.
The DM chooses whether to revise her guess. To do so, she must infer
$p(\theta\mid d,h)$. 

We embed the problem in our framework, and compare the DM's behavior
under limited-propagation beliefs to the Bayesian rational benchmark.
We capture limited inference capabilities by assuming that flows consist
of single-link paths from the givens $d$ and $h$. Then, $F_{d\rightarrow\theta}=\theta\rightarrow d$
and $F_{h\rightarrow\theta}=\emptyset$. The order is irrelevant because
$N-G$ is a singleton. Our story behind this specification is not
necessarily that longer chains are hard as such. Rather, as argued
by \citet[Ch. 6]{pearl2018book} in the context of the Monty Hall
problem, reasoning over a $v$-collider may be difficult for people.

Applying Equation (\ref{eq: subjective beliefs}), $\hat{p}\left(x_{\theta}\mid x_{h},x_{d}\right)=p\left(x_{\theta}\mid x_{d}\right).$
In words, the DM departs from correct updating by failing to take
into account her own initial guess $h$ when drawing the inference
from $d$. Given the usual parameterization of the Monty Hall problem
($\theta,h\thicksim U\{1,2,3\}$), and $d\mid(\theta,h)\thicksim U\left(\{1,2,3\}-\{\theta,h\}\right)$),
$p(\theta\mid d)$ coincides with the prior $p(\theta)$, and the
DM would not strictly prefer to revise her guess, unlike a Bayesian.

\subsubsection{Upstream and downstream inferences}

Suppose the true data-generating process obeys the following DAG $R$:
\[
\begin{array}{ccccccccc}
 &  & 2 &  &  &  & 5\\
 & \nearrow &  & \searrow &  & \nearrow &  & \searrow\\
1 &  &  &  & 4 &  &  &  & 7\\
 & \searrow &  & \nearrow &  & \searrow &  & \nearrow\\
 &  & 3 &  &  &  & 6
\end{array}
\]

Let $G=\{4\}$. Assume that for every distinct nodes $i,j$, $F_{i\rightarrow j}=\left\{ (i,j),(j,i)\right\} \cap R$.
That is, flows consist of all paths of length $1$, without making
a distinction between given and hypothetical variables. The order
satisfies $1\succ2,3\succ4$ and $7\succ5,6\succ4$.

Let us construct the limited-propagation belief. First, for every
node $i$, $R_{i}$ consists of all the edges it is part of. Second,
for every $i=2,3,5,6$, $M(i)=\{4\}$, while $M(1)=\{2,3\}$ and $M(7)=\{5,6\}$.
Then,
\[
\hat{p}\left(\cdot\mid x_{4}\right)=p\left(x_{1}\mid x_{2},x_{3}\right)p\left(x_{7}\mid x_{5},x_{6}\right)\prod_{i=2,3,5,6}p\left(x_{i}\mid x_{4}\right)
\]
That is, the DM updates her belief as if she held the subjective DAG
that deviates from $R$ by inverting the links on its left part. Thus,
although the rule for flows does not explicitly discriminate between
upstream and downstream links, the directional asymmetry inherent
in $R$ implies that downstream inferences are correct while upstream
inferences are wrong.

\subsubsection{A multi-DAG representation\label{subsec:A-multi-DAG-representation}}

In all previous examples, limited-propagation beliefs were effectively
consistent with a DM who fits the unconditional distribution $p$
to a single subjective DAG, and then conditions the resulting subjective
probability on the givens. We present an example in which this is
not the case.

Suppose that the true DAG $R$ is
\[
\begin{array}{ccccc}
1 & \leftarrow & 3\\
 &  & \downarrow & \searrow\\
2 & \rightarrow & 4 & \rightarrow & 5
\end{array}.
\]
Let $G=\{1,2\}$. For any $j>2$, let $F_{1\rightarrow j}$ and $F_{2\rightarrow j}$
consist of the shortest paths in $R$ from $1$ or $2$ to $j$. Let
$F_{4\rightarrow5}=4\rightarrow5$, and $F_{i\rightarrow j}=\emptyset$
for all other combinations. Then, $R_{3}$ and $R_{4}$ are $R$ restricted
to $\{1,2,3,4\}$, while $R_{5}$ is $R$ minus the link $3\rightarrow4$.
The order $\succ$ agrees with the natural order $>$. Then, $M(3)=M(4)=\{1,2\}$,
and $M(5)=\{1,2,4\}$. 

The DM's limited-propagation belief is thus
\[
\hat{p}\left(x_{3},x_{4},x_{5}\mid x_{1},x_{2}\right)=p\left(x_{3}\mid x_{1},x_{2}\right)p\left(x_{4}\mid x_{1},x_{2}\right)p_{R_{5}}\left(x_{5}\mid x_{1},x_{2},x_{4}\right)
\]
This belief can be rewritten as
\begin{equation}
\hat{p}\left(x_{3},x_{4},x_{5}\mid x_{1},x_{2}\right)=p\left(x_{3}\mid x_{1},x_{2}\right)p\left(x_{4}\mid x_{1},x_{2}\right)\sum_{k}p\left(X_{3}=k\mid x_{1}\right)p\left(x_{5}\mid X_{3}=k,x_{2}\right)\label{eq: formula public good}
\end{equation}
The belief $\hat{p}(\cdot\mid x_{1},x_{2})$ does not reduce to $p_{Q}(\cdot\mid x_{1},x_{2})$
for any DAG $Q$. Clearly, $\hat{p}(\cdot\mid x_{1},x_{2})\neq p(\cdot\mid x_{1},x_{2})$,
so $Q=R$ will not work. One other natural candidate would be the
DAG $Q=R_{5}$. If $\hat{p}(\cdot\mid x_{1},x_{2})=p_{Q}(\cdot\mid x_{1},x_{2})$,
then $5$ would be independent of $1$ given $3$ and $4$, which
it is not. Proposition \ref{prop: no DAG example} makes this point
more explicitly. Section \ref{example public good} presents an application
based on this example.

\section{Evidence on contingent thinking\label{sec:Limited-propagation-and}}

In this section we demonstrate how the framework can organize and
shed light on two sets of experiments studying settings that require
contingent thinking. We then use the recreational puzzle Kakuro to
suggest a new direction for experimental work on this topic.

The first set of experiments includes Monty Hall (MH) and Acquire
a Company (AaC). \citet{friedman1998monty} conducts an experiment
on the Monty Hall problem described above, and finds that most subjects
do not switch, even though switching strictly raises their probability
of winning. In AaC, a company has value $v\in\left\{ v_{L},v_{H}\right\} $
with equal probability, where $v_{L}<v_{H}<2v_{L}$. The company knows
$v$, the subject values the company at $\frac{3}{2}v$, does not
know $v$, knows that the company knows $v$, and offers a price.
The company accepts if and only if the price weakly exceeds $v$.
\citet{samuelson1985winner} and \citet{charness2009origin} find
that most subjects offer more than $v_{L}$, the optimal strategy
for non-risk-lovers.

Both MH and AaC have data-generating processes with the same underlying
DAG, namely a $v$-collider $1\rightarrow3\leftarrow2.$ In MH, $1$
is the door picked, $2$ is the location of the prize, and $3$ is
the door opened. In AaC, $1$ is the price offered, $2$ is the value
of the company, and $3$ is whether the company accepts. A key difference
is that the subject's relevant choice variable belongs to the DAG
in AaC (the offered price is node $1$) but lies outside it in MH,
where the initial pick is node $1$ but the choice that actually matters,
whether to switch, is not represented in the DAG. With no contingent
thinking and when flows from givens cover all paths of length $2$,
limited-propagation beliefs make correct predictions about both whether
to switch and which price to offer. They make incorrect predictions
when they leave out the path from $1$ to $2$, as when they cover
only paths of length $1$.

\citet{martinez2019failures} study a variant of AaC in which there
are two companies---one \textit{known} to be of high value and the
other of low value---to which the subject offers a single common
price. This leads to the DAG $2\leftarrow1\rightarrow3,$ where $1$
again represents the price offered while $2$ and $3$ now represent
the acceptance decisions of the two companies. In contrast to AaC,
limited-propagation beliefs make correct predictions as long as flows
from the given $1$ contain the edges from it to $2$ and $3$, since
no inference through a collider is required.

The second set of experiments includes auctions \citep{kagel1987information},
strategic voting \citep{esponda2014hypothetical,esponda2024contingent},
learning from prices \citep{ngangoue2021learning}, and dynamic public
good provision \citep{calford2024contingent}. In each, choosing correctly
requires the subject to extrapolate from a variable that depends on
other players' actions (the winning bid, election winner, market price,
or whether others contributed) to the others' private information.
Behavior moves toward the rational benchmark when the relevant variable
is observed before the subject acts, and away from it when that variable
is hypothetical. \citet{kagel1987information} document distinct behavior
in second-price auctions and in strategically equivalent English auctions,
and \citet{levin1996Revenue} show that first-price auctions raise
more revenue than English auctions in a common-values setting when
bidders suffer from the winner's curse. \citet{esponda2014hypothetical}
show that voters in a committee largely fail to condition on being
pivotal but improve substantially when the others' votes are observed
first. \citet{ngangoue2021learning} find that traders extract information
from the realized asset price but not from the hypothetical price
when that price depends on the other trader's order. \citet{calford2024contingent}
compare static and dynamic public goods provision, and find that the
provision rate is closer to the Nash equilibrium rate in the latter
setting.

All four classes of experiments correspond to models with DAGs of
the form

\[
3\leftarrow1\rightarrow4\rightarrow5\leftarrow2
\]
where $2$ and $3$ are givens, with minor variations.\footnote{In \citet{ngangoue2021learning}, there is an additional edge $4\rightarrow2$.
The three-player setting of \citet{calford2024contingent} requires
$4$ to be the pair of actions of the other two players, or to incorporate
an additional path $1\rightarrow6\rightarrow5$, where $6$ is player
$3$'s contribution; our static framework necessarily abstracts from
some of the dynamic considerations important for their identification
strategy. In auctions where $1$ is the type of the other player,
the edge $3\leftarrow1$ can be dropped. None affect the below discussion.
In all except \citet{ngangoue2021learning}, one can also include
the edge $3\rightarrow2$. Since $\{3,2\}\subset G$ and $3$ blocks
any path from 2 that uses this link, this is immaterial.} The interpretation in auction is that $1$ is the type of the other
player or the common value, $4$ the bid of the other, $3$ the type
of the subject, $2$ the action of the subject, and $5$ the allocation/transfer.
The interpretation in jury voting is that $1$ is a common value,
$4$ the vote of the other, $3$ the signal of the subject (which
Esponda and Vespa leave out to make their point more cleanly), $2$
the vote of the subject, and $5$ the election winner. The interpretation
in markets is that $1$ is the asset's common value, $4$ the action
of the other, $3$ the signal of the subject, $2$ the price faced
by the other player, and $5$ the price faced by the subject. The
interpretation in public good game is that $1$ the value of the public
good to the subject, $3$ is the subject's signal about its value,
$4$ the contribution of the other(s), $5$ whether the public good
is provided, and $2$ the contribution of the subject.

In each case, the subject needs to form beliefs about the joint distribution
of $X_{1}$ and $X_{5}$ (as well as $X_{4}$ in a second-price auction)
in order to choose the correct action. Consider limited-propagation
beliefs with no contingent thinking and flows from givens that contain
only paths of length 2 or less. The order is immaterial. Then,
\begin{align*}
\hat{p}\left(x_{1},x_{4},x_{5}|x_{2},x_{3}\right) & =p\left(x_{5}|x_{2}\right)p\left(x_{4}|x_{2},x_{3}\right)p\left(x_{1}|x_{3}\right)\\
\hat{p}\left(x_{1},x_{5}|x_{2},x_{3},x_{4}\right) & =p\left(x_{5}|x_{2},x_{4}\right)p\left(x_{1}|x_{3},x_{4}\right)=p\left(x_{1},x_{5}|x_{2},x_{3},x_{4}\right)
\end{align*}
The beliefs about $X_{1}$ and $X_{5}$ are incorrect when the givens
are $G=\{2,3\}$ because they understate the correlation between $1$
and $5$ implied by the path $1\rightarrow4\rightarrow5$. However,
they are correct when the givens are $G^{\prime}=\{2,3,4\}$ (or $G'=\left\{ 2,3,5\right\} $
in the market experiment). Explicitly observing the bid, vote, price,
or contribution causes the subject to incorporate the otherwise-omitted
paths.

This pattern accords with the experimental evidence: across all four
settings, subjects' play moves closer to the rational benchmark in
treatments where the other player's action is observed rather than
hypothetical. For example, Section IV of \citet{li2017obviously}
shows that play is closer to rational in ascending auctions than in
strategically equivalent sealed-bid second-price auctions. In the
present framework, the former corresponds to limited-propagation beliefs
with $4$ (the other's bid) included in the givens, while the latter
corresponds to having only $2$ and $3$ given.

\subsection{New Experiments? Kakuro\label{subsec:Kakuro}}

In this subsection we propose a direction for new experiments on contingent
thinking, using the recreational puzzle game Kakuro, also called Cross-Sum,
as a template. The player's objective in this game is to fill a crossword-puzzle-like
grid with numbers (Figure \ref{fig:kakuro}). As in a crossword, each
row or column corresponds to a clue, which indicates the sum of the
numbers in that row or column.\footnote{We ignore other rules, such as that the numbers cannot repeat within
a row/column.} The puzzles typically have a unique solution, logically deducible
from the clues alone. Empirically, difficulty varies across puzzles
and depends, among other things, on whether solving requires contingent
reasoning or long chains of interdependent thinking. We will show
how the model enables us to capture these gradations.

\begin{figure}
\caption{\label{fig:kakuro}A Kakuro puzzle}

\includegraphics[scale=0.4]{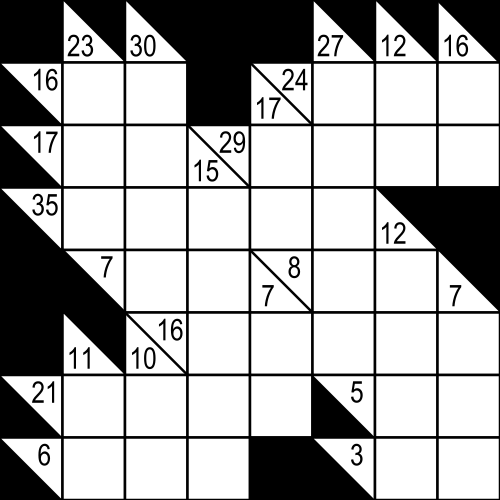}

(c) 2007 by Octahedron80, \url{https://commons.wikimedia.org/wiki/File:Kakuro_black_box.svg}
\end{figure}

Consider a simple Kakuro grid with one non-trivial row and one non-trivial
column:
\[
\begin{array}{ccc}
 & c_{1} & c_{2}\\
r_{1} & s_{11} & s_{12}\\
r_{2} & s_{21} & \blacksquare
\end{array}
\]
The variable $s_{ij}$ corresponds to the solution to the cell in
row $i$ and column $j$; $r_{i}$ to the sum of the values in row
$i$; and $c_{j}$ to the sum of the values in column $j$. Assume
that $r_{1}$ and $c_{1}$ are known and observed by the DM, that
$r_{2}$ and $c_{2}$ are unobserved, and that $s_{11},s_{12},s_{22}$
are mutually independent. The DM's task is to make a guess about the
values of the $s_{ij}$ variables.

Consider a hint that reveals the value of one of the cells. Knowing
the value of any of the cells should be equally valuable to a Bayesian.
Combined with the clues $r_{1}$ and $c_{1}$, learning any $s_{ij}$
leaves a system of two equations with two unknowns: $s_{11}+s_{12}=r_{1}$
and $s_{21}+s_{11}=c_{1}$. This leaves no residual uncertainty. However,
the hint ``$x_{11}=k$'' is intuitively more valuable than the hint
``$x_{21}=k$'' or ``$x_{12}=k$'' because using it requires fewer
steps of contingent reasoning.

Under the same specification of flows as in our auction example, limited-propagation
beliefs make this prediction. Notice that $p$ is consistent with
DAG 
\[
s_{21}\rightarrow c_{1}\leftarrow s_{11}\rightarrow r_{1}\leftarrow s_{12}.
\]
Again, assume that flows consist of paths no longer than $2$ from
given nodes, and are empty for hypothetical nodes, and fix any admissible
order. 

First, consider the hint ``$s_{11}=k$.'' In this case, $G=\{c_{1},r_{1},s_{11}\}$;
$F_{s_{11}\rightarrow s_{21}}=s_{21}\rightarrow c_{1}\leftarrow s_{11}$;
$F_{c_{1}\rightarrow s_{21}}=s_{21}\rightarrow c_{1}$; and there
are no other flows into $s_{21}$. Therefore, $R_{s_{21}}$ is $s_{21}\rightarrow c_{1}\leftarrow s_{11}$,
which induces the correct conditional belief that $s_{21}=c_{1}-k$
for sure. By the same logic, the DM believes correctly that $s_{12}=r_{1}-k$
for sure.

Second, consider the hint ``$s_{12}=k$.'' In this case, $G=\{c_{1},r_{1},s_{12}\}$,
so $s_{11}$ is now a hypothetical variable, not a given one. This
implies that $F_{c_{1}\rightarrow s_{21}}=s_{21}\rightarrow c_{1}$
and $F_{j\rightarrow s_{21}}=\emptyset$ for all $j\neq c_{1}$, since
inference about $s_{21}$ from $r_{1}$ or $s_{12}$ requires more
than two steps, and since there are no inferences from the hypothetical
variable $s_{11}$. Therefore, the DM updates her beliefs about $s_{21}$
to $p(s_{21}\mid c_{1})$; her beliefs do not depend on either $r_{1}$
or the hint $s_{12}$. Even though she updates correctly about $s_{11}$,
she does not about the variable that requires more steps of inference.
Consequently, the hint ``$s_{12}=k$'' is \textit{less valuable}
than the hint ``$s_{11}=k$.''

Extending this analysis to a probabilistic setting with no certain
logical implications is straightforward---unlike other models of
deviations from logical omniscience \citep[e.g.,][]{lipman1999decision,jakobsen2020complex,piermont2026failures}.
This is easily seen by either adding noise to the summation formulas
for $r_{1}$ and $c_{1}$ (e.g., $r_{1}=s_{11}+s_{12}+\varepsilon$),
or by expanding the rows and columns so that each include an extra
cell (with values $s_{31}$ and $s_{13}$). In either case, Bayesian
beliefs are non-degenerate even after the hint. Nevertheless, limited-propagation
beliefs induced by the above rules for flows still imply that the
second hint is less valuable than the first. For example, under the
extra-cells extension,
\[
\hat{p}(s_{21},s_{12},s_{13},s_{31}\mid s_{11},r_{1},c_{1})=p(s_{21}\mid s_{11},c_{1})p(s_{12}\mid s_{11},r_{1})p(s_{13}\mid s_{11},r_{1})p(s_{31}\mid s_{11},c_{1})
\]
whereas
\[
\hat{p}(s_{21},s_{11},s_{13},s_{31}\mid s_{12},r_{1},c_{1})=p(s_{21}\mid c_{1})p(s_{11}\mid s_{12},r_{1})p(s_{31}\mid c_{1})p(s_{13}\mid s_{12},r_{1}).
\]
Both formulas partially neglect the correlation between $s_{21}$
and $s_{31}$, yet under an ex-ante symmetric prior, conditioning
on $s_{11}$ reduces uncertainty more than conditioning on $s_{12}$.
Thus, a single mechanism accounts for both absence of logical omniscience
and probabilistic mistakes.

Both deterministic and probabilistic versions of Kakuro provide a
possible template for future experiments on contingent thinking. Under
natural structure on the flows, they generate predictions regarding
the comparative difficulty of different puzzles and the value of various
``hints.'' As we show in Section \ref{subsec:social-learning},
data-generating processes having similar causal structure lend themselves
to economic applications.

\section{Results \label{sec:general analysis}}

In this section, we collect general properties of limited-propagation
beliefs. We begin by exploring the connection between limited-propagation
beliefs and other models of boundedly rational belief formation. We
then provide necessary and sufficient conditions for limited-propagation
beliefs to be correct. Next, we explore how our framework captures
limited contingent thinking. In our model, absence of contingent thinking
implies correlation neglect: variables that should be correlated conditional
on the givens are perceived as independent. Moreover, beliefs violate
iterated expectations, which can lead to failures of dynamic consistency
in a decision tree. We conclude the section with an exploration of
the representation's testable implications.

\subsection{Relationship with other non-Bayesian beliefs\label{subsec:Relationship-with-other}}

\citet{spiegler2016bayesian} argues that several existing models
of boundedly rational inference, including \citet{eyster2005cursed},
\citet{jehiel2008revisiting} and \citet{esponda2008behavioral},
can be represented as if updating a subjective belief that factorizes
the objective distribution according to a fixed misspecified DAG.
For instance, fully cursed inference by player 1 links all other player's
actions to player 1's type instead of their own. We say that limited-propagation
beliefs (under some specification of $(G,F,\succ)$) have a \emph{single-DAG
representation} if there is a DAG $Q$ such that $\hat{p}(x_{N-G}\mid x_{G})=p_{Q}(x_{N-G}\mid x_{G})$
for every $x_{N}$.

To see that limited-propagation beliefs cannot always be represented
as the factorization according to some fixed DAG, recall the example
in Section \ref{subsec:A-multi-DAG-representation}.
\begin{prop}
\label{prop: no DAG example}With the givens and flows in Section
\ref{subsec:A-multi-DAG-representation}, limited-propagation beliefs
lack a single-DAG representation. That is, for any DAG $Q$ there
exists $p$ that is consistent with $R$, such that $\hat{p}\left(x_{3},x_{4},x_{5}\mid x_{1},x_{2}\right)\neq p_{Q}\left(x_{3},x_{4},x_{5}\mid x_{1},x_{2}\right)$
for some $x$.
\end{prop}
The result will follow from the application in Section \ref{example public good},
especially Proposition \ref{prop:public good}. A formal proof is
thus omitted. For an intuition, consider a distribution $p$ so that
$X_{5}=1$ only if $X_{3}=0$ and $X_{4}=1$, while $X_{4}=1$ only
if $X_{3}=1$ and $X_{2}=1$. Limited-propagation beliefs may assign
positive probability to $\left(X_{3},X_{4},X_{5}\right)=\left(1,1,1\right)$
because they propagate information to $X_{5}$ from $X_{3}$ only
indirectly, through the flow from 1 to $5$. Therefore, they may neglect
the negative correlation between $X_{3}$ and $X_{5}$ given $X_{4}$.
However, single-DAG beliefs never will give positive probability to
that event. Intuitively, if both $3$ and $4$ cause 5, then the DM
correctly updates $5$ conditional $3$ and $4$ and rules out all
equal being to $1$. If she leaves out either 3 or 4 from the causes
of 5, then she infers the probability of the omitted causes and rules
out $X_{5}=1$. For example, if $3$ but not 4 causes 5, then she
correctly updates that $X_{5}=0$ when $X_{3}=0$, since $X_{4}=0$
whenever $X_{3}=0$. 

We now ask under what conditions limited-propagation beliefs have
a single DAG representation. To state the result, let $F^{*}_{i\rightarrow j}$
be the set of all $R$-paths from $i$ to $j$, and say that flows
satisfy the \emph{all-or-nothing }property if for every $i\neq j$,
$F_{i\rightarrow j}\in\{F^{*}_{i\rightarrow j},\emptyset\}$. The
property means that if information flows from $i$ to $j$ at all,
then it travels through \textit{all} the paths in $R$ from $i$ to
$j$. 
\begin{prop}
\label{prop: polytree} If flows satisfy the all-or-nothing property,
then $\hat{p}(x_{N-G}\mid x_{G})\equiv p_{Q}(x_{N-G}\mid x_{G})$,
where $Q$ is the DAG defined by $Q(i)\equiv M(i)$ for every i.
\end{prop}
In many of our examples, the true DAG $R$ is a \emph{polytree}: there
is at most one path between any pair of nodes. Any specification of
the flows for a polytree satisfies the all-or-nothing property. Consequently,
limited-propagation beliefs have a single DAG representation in those
examples. Crucially, flows do \textit{not} satisfy the all-or-nothing
property in the example of Section \ref{subsec:A-multi-DAG-representation};
e.g., the flow from $1$ to $5$ leaves out the path $1\rightarrow3\rightarrow4\rightarrow5$.
\begin{proof}[Proof of Proposition \ref{prop: polytree}]
We begin with two observations that hold for any $F$. First, the
directed graph defined by $Q(i)\equiv M(i)$ is acyclic, since $j\in M(i)$
only if $i\succ j$. Second, $M(i)\cup\{i\}\subseteq N(R_{i})$, where
$N(R_{i})$ is the set of nodes in $R_{i}$. Therefore, if we show
that $p_{R_{i}}(x_{N(R_{i})})=p(x_{N(R_{i})})$ for every $i\in N-G$,
it will immediately follow that $p_{R_{i}}(x_{i}\mid x_{M(i)})=p(x_{i}\mid x_{M(i)})$.

Pick any $i\in N-G$, and WLOG, reorder all nodes so that $N^{*}\equiv\left\{ 1,\dots,m\right\} =N(R_{i})$
and for each $j\leq m$, $\{1,\dots,j-1\}$ contains no descendants
of $j$. We show that $R(j)\cap\{1,\dots,j\}=R_{i}(j)$ for any $j$.
First, if $k\in R_{i}(j)$, then $kRj$ and $k$ is an ancestor of
$j$. Therefore, $R_{i}(j)\subset R(j)\cap\{1,\dots,j\}$. Second,
if $k\in R(j)\cap\{1,\dots,j\}$, then by construction, there exists
$j^{*}$ so that $k$ is in a path in $F_{j^{*}\rightarrow i}$ .
Since $j\in N^{*}$, there is a path in $R$ from $j$ to $i$. Then,
the path from $j^{*}$ to $k$ combined with $(k,j)$ and the path
from $j$ to $i$ yields a path from $j^{*}$ to $i$. This path must
be contained in $F_{j^{*}\rightarrow i}$ by the all-or-nothing property,
so $(k,j)\in R_{i}$.

Clearly $p_{R_{i}}(x_{1})=p(x_{1})$. Assume that $p(x_{1},\dots,x_{j-1})=p_{R_{i}}(x_{1},\dots,x_{j-1})$.
Let $L=R(j)-R_{i}(j)$ be the latent parents of $j$. Then, 
\begin{align*}
p\left(x_{1},\dots,x_{j}\right)= & \sum_{x_{L}}p\left(x_{j}|x_{R(j)}\right)p\left(x_{1},\dots,x_{j-1},x_{L}\right)\\
= & \sum_{x_{L}}p\left(x_{j}|x_{R_{i}(j)},x_{L}\right)p\left(x_{1},\dots,x_{j-1},x_{L}\right)
\end{align*}
We claim that $X_{L}\perp X_{\{1,\dots,j-1\}}$. If not, then they
have a common ancestor $k$, and given how $N^{*}$ is ordered, $k\notin\left\{ j,\dots,m^{*}\right\} $.
Let $j^{\prime}<j$ and $l\in L$ be descendants of $k$, and pick
$j'$ to be the first such descendant. Then, there is a path in $R$
between $j'<j$ and $l\in L$ that goes through $k$ and does not
intersect $\left\{ j,\dots,m^{*}\right\} $. By construction, there
exists $j^{*}$ so that $j^{\prime}$ lies on a path in $F_{j^{*}\rightarrow i}$
. Create a new path from $j^{*}$ to $i$, which combines that path
from $j^{*}$ to $j^{\prime}$, the path from $j^{\prime}$ to $l$
that does not include $j$, the link $(l,j)$, and a path from $j$
to $i$. This constructed path from $j^{*}$ to $i$ includes $l$.
By the all-or-nothing property, it must be contained in $F_{j^{*}\rightarrow i}$,
contradicting $l\in L$. 

Given the independence, $p\left(x_{1},\dots,x_{j-1},x_{L}\right)=p\left(x_{1},\dots,x_{j-1}\right)p\left(x_{L}\right)$,
so
\begin{align*}
p\left(x_{1},\dots,x_{j}\right)= & \left(\sum_{x_{L}}p\left(x_{j}|x_{R_{i}(j)},x_{L}\right)p\left(x_{L}\right)\right)p\left(x_{\{1,\dots,j-1\}}\right)\\
= & p\left(x_{j}|x_{R_{i}(j)}\right)p_{R_{i}}\left(x_{\{1,\dots,j-1\}}\right)=p_{R_{i}}\left(x_{\{1,\dots,j\}}\right).
\end{align*}
Inductively, $p_{R_{i}}\left(x_{N^{*}}\right)=p\left(x_{N^{*}}\right)$.
Therefore, $p_{R_{i}}\left(x_{N}\right)=p\left(x_{N^{*}}\right)\prod_{i\notin N}p(x_{i})$,
since $R_{i}(i)=\emptyset$ for all $i\notin N^{*}$. 
\end{proof}

\subsection{Are richer flows always better?}

We leverage Proposition \ref{prop: polytree} to show by example that
more flows need not lead to better decisions. Suppose that a DM with
limited-propagation beliefs chooses an action to maximize expected
utility with expectation taken with respect to her limited-propagation
beliefs. The action is not represented by a node in the DAG. An outside
observer evaluates the DM's decision quality according to her ex-ante
expected utility according to the true distribution $p$.

Let $N=\{1,2,3,4\}$ and $G=\emptyset$. The true DAG $R$ is complete,
such that $(i,j)\in R$ if and only if $i<j$. Suppose $4\succ3\succ2\succ1$,
and compare two flow specifications, $F^{1}$ and $F^{2}$. Under
$F^{1}$, there no flows that involve node $3$; $F^{1}_{1\rightarrow2}=1\rightarrow2$;
$F^{1}_{2\rightarrow4}=2\rightarrow4$; and $F^{1}_{1\rightarrow4}=1\rightarrow2\rightarrow4$.
This specification induces a single-DAG representation, where the
DAG $R^{1}$ coincides with $F^{1}_{1\rightarrow4}=1\rightarrow2\rightarrow4$.
The specification $F^{2}$ \textit{adds} the path $1\rightarrow3\rightarrow4$
to $F_{1\rightarrow4}$, and also adds the paths $F^{2}_{1\rightarrow3}=1\rightarrow3$
and $F^{2}_{3\rightarrow4}=3\rightarrow4$. This specification, too,
induces a single-DAG representation, where the DAG $R^{2}$ adds the
path $1\rightarrow3\rightarrow4$ to $R^{1}$.

Thus, under both specifications, the DM chooses her action by fitting
a subjective DAG to the true distribution $p$. While $R^{2}$ contains
$R^{1}$, neither DAG is correctly specified because both are strict
sub-graphs of $R$. Proposition 10 in \citet{spiegler2016bayesian}
shows that in such a scenario, there are objective distributions for
which $R^{1}$ will lead to higher objective expected payoffs than
$R^{2}$. Thus, the outside observer would not necessarily say that
the DM with the richer $F^{2}$ made higher quality decisions than
she would with $F^{1}$.\footnote{We can also evaluate flows by the accuracy of the beliefs they generate,
measured by, say, the Kullback-Leibler divergence with respect to
the true distribution. It is possible to construct an example of a
true distribution for which $F^{2}$ will lead to larger divergence
than $F^{1}$.}

\subsection{Rationality }

We now consider when limited-propagation beliefs agree with Bayesian
updating. The result provides a necessary condition and a slightly
stronger sufficient condition for local propagation beliefs to agree
with Bayes rule. To state it, say that the order $\succ$ \emph{extends}
$R$ if $i\succ j$ implies that $i$ is not an ancestor of $j$ according
to $R$. 
\begin{prop}
\label{prop:flows that imply bayes}(i) If there exist $i,j\notin G$
such that $(i,j)\in R-\left(F_{i\rightarrow j}\cup F_{j\rightarrow i}\right)$,
then $\hat{p}(\cdot\mid x_{G})\neq p(\cdot\mid x_{G})$ for almost
all $p$ and any order. (ii) If $\succ$ extends $R$ and $(i,j)\in F_{i\rightarrow j}$
for every $(i,j)\in R$ with $j\notin G$, then $\hat{p}(\cdot\mid x_{G})=p(\cdot\mid x_{G})$
for every $p$. 
\end{prop}
Some contingent thinking is necessary for rationality whenever updating
about more than one variable. Otherwise, the DM fails to take into
account the relationships between hypotheticals. The proposition shows
that neglecting a direct link between two hypotheticals leads to a
failure of Bayesianism, for any $p$ for which that link matters (which
is ``generically'' the case).

Conversely, sufficient contingent thinking implies Bayesian updating.
Specifically, if each variable flows to its children, then there is
an order for which the limited-propagation belief agrees with Bayesianism.
This order extends $R$. Since $G$ must be ordered first, this means
that $G$ is ancestral with respect to $R$: for every $g\in G$,
$R(g)\subseteq G$. With some contingent thinking, the DM need not
consider all possible paths through which one variable might affect
another. It is sufficient that she only considers each variable's
direct effects, provided that these are then propagated onward to
the affected variables' descendants. 
\begin{proof}[Proof of Proposition \ref{prop:flows that imply bayes}]
\textit{(i)} Suppose that there exist $i,j\notin G$ such that $(i,j)\in R-\left(F_{i\rightarrow j}\cup F_{j\rightarrow i}\right)$.
Consider $p$ so that 
\[
p\left(x_{1},\dots,x_{N}\right)=p\left(x_{i},x_{j}\right)\prod_{k\neq i,j}p\left(x_{k}\right)
\]
for all $x$, and where $j\not\perp_{p}i$. For any DAG $Q$, 
\[
p_{Q}\left(x_{1},\dots,x_{N}\right)=p\left(x_{i}\mid x_{Q(i)\cap\{j\}}\right)p\left(x_{j}\mid x_{Q(j)\cap\{i\}}\right)\prod_{k\neq i,j}p\left(x_{k}\right).
\]
In particular, for $Q=R_{i}$ or $Q=R_{j}$, $p_{Q}\left(x_{1},\dots,x_{N}\right)=\prod_{k}p\left(x_{k}\right)$.
Therefore, 
\[
\hat{p}(x_{i},x_{j}|x_{G})=p(x_{i})p(x_{j})\neq p(x_{i},x_{j})=p(x_{i},x_{j}|x_{G}).
\]

Using \citet{caron2005zero}, we can extend the inequality $\hat{p}(x_{i},x_{j}|x_{G})\neq p(x_{i},x_{j}|x_{G})$
to almost all $p$. Note that $p$ can be viewed as an element of
$\mathbb{R}^{m}$ for some $m$. The function $f(p)=\sum_{x_{N}}\left(\hat{p}(x_{N-G}\mid x_{G})-p(x_{N-G}\mid x_{G})\right)^{2}$
has the same zeros as a polynomial on $\mathbb{R}^{m}$, and $f(p)$
equals zero only if $\hat{p}(x_{N-G}|x_{G})=p(x_{N-G}|x_{G})$ for
every $x_{N}$. Caron and Traynor show that every non-zero polynomial
on $\mathbb{R}^{m}$ is non-zero almost everywhere, so almost all
$p$ satisfy $\hat{p}(\cdot|x_{G})\neq p(\cdot|x_{G})$.

\textit{(ii)} Suppose that $\succ$ extends $R$, and that for every
$(i,j)\in R$ with $j\notin G$, $(i,j)\in F_{i\rightarrow j}$. Then,
$R(j)\subset M(j)$ and $R_{j}(j)=R(j)$. Therefore, $p_{R_{j}}\left(x_{j}|x_{M(j)}\right)=p\left(x_{j}|\text{\ensuremath{x_{R(j)}}}\right)$
for all $j\in N-G$. Conclude 
\[
\hat{p}\left(x_{N-G}|x_{G}\right)=\prod_{j\notin G}p\left(x_{j}|\text{\ensuremath{x_{R(j)}}}\right)=p\left(x_{N-G}|x_{G}\right)
\]
since $p=p_{R}$ by assumption, establishing the result. 
\end{proof}

\subsection{Failures of contingent thinking\label{subsec: contingent thinking}}

For this subsection, we consider two sets of givens, $G$ and $G'=G\cup\{g^{*}\}$,
and compare the resulting limited-propagation beliefs. Conditioning
on more variables may change the implications that the DM considers.
To accommodate this, we index the flows by the set of givens, i.e.
$F^{H}_{i\rightarrow j}$ is flows from $i$ to $j$ when the givens
are $H\in\left\{ G,G^{\prime}\right\} $. The order $\succ$ is fixed
throughout. 

We say that flows \emph{exhibit no contingent thinking} if, for every
$H\in\{G,G^{\prime}\}$, $F^{H}_{i\rightarrow j}=\emptyset$ whenever
$i\notin H$. This is a strong definition that requires the DM to
take no inference steps from non-givens. It provides a clean and natural
benchmark to which one can compare full rationality and intermediate
degrees of contingent thinking.

Failure of contingent thinking immediately implies correlation neglect.
\begin{prop}
\label{prop: contingent think correlation neglect}If flows exhibit
no contingent thinking, then limited-propagation beliefs regard $i$
and $j$ as independent conditional on $G$, for every distinct $i,j\in N-G$.
\end{prop}
\begin{proof}
This follows from Equation (\ref{eq: subjective beliefs}) and observing
that with no contingent thinking, $M(i),M(j)\subset G$.
\end{proof}

That is, the DM neglects any correlation between $i$ and $j$ not
picked up in the information $G$. Correlation neglect has been documented
experimentally support \citep[e.g, ][]{EysterWeizsacker2010,SalantRubinstein2015,enkezimmerman2013}
and have been studied theoretically \citep[e.g., ][]{OrtolevaSnowberg2015,LevyRazin2015QJPS,EllisPiccione2017}.
The result thus provides a novel link between failures of contingent
thinking and correlation neglect. 

Second, we are interested in comparing $\hat{p}(x_{N-G^{\prime}}\mid x_{G})$
and $\hat{p}(x_{N-G^{\prime}}\mid x_{G'})$. The comparison makes
sense under the following assumptions.
\begin{assumption*}
\noindent\label{assumption}The following hold.

\noindent (1) $F^{G}_{i\rightarrow j}=F^{G'}_{i\rightarrow j}$ for
all $i\notin g^{*}$.

\noindent (2) Let $H\in\{G,G^{\prime}\}$, and suppose that $i,j\in H$
and that $F^{H}_{j\rightarrow k}$ contains a path from $i$ to $k$.
Then, this path also belongs to $F^{H}_{i\rightarrow k}$.
\end{assumption*}
Part (1) says that flows from $i$ only depend on whether or not $i$
is a given. Part (2) says that if two variables are both givens, then
a stronger version of the betweenness property must hold. Specifically,
if $i$ lies on a path from $j$ to $k$ that belongs to $F^{H}_{j\rightarrow k}$,
then the sub-path from $i$ to $k$ must belong to $F^{H}_{i\rightarrow k}$.
Recall that we emphatically did not impose this property on two arbitrary
nodes $i$ and $j$, usually to allow for cases in which $j$ is a
given while $i$ is a hypothetical, and the DM is better at drawing
inferences from givens. However, when $i$ and $j$ are both givens,
the property is defensible. Together, the two parts rule out ``discovering''
an additional channel through which a given influences a variable
when a new given is added.

We now explore whether limited-propagation beliefs satisfy the standard
iterated expectations property.
\begin{defn}[Iterated expectations]
\label{def:martingale}Fix $G$ and $G'$. Limited-propagation beliefs
satisfy iterated expectations if for every objective distribution
$p$ and every $x_{G},$$x_{N-G'}$,
\begin{equation}
\hat{p}\left(x_{N-G'}\mid x_{G}\right)=\sum_{x_{g^{*}}}\hat{p}(x_{g^{*}}\mid x_{G})\hat{p}\left(x_{N-G'}\mid x_{G},x_{g^{*}}\right).\label{eq:iterated expectations}
\end{equation}
\end{defn}
In other words, if we regard $\hat{p}(x_{N-G'}\mid x_{G})$ as the
DM's prior belief over $x_{N-G'}$ and $x_{g^{*}}$ as her information,
then (\ref{eq:iterated expectations}) says that on average, the DM's
posterior belief is equal to her prior. \citet{spiegler2020can} asks
a similar question for single-DAG belief models, mainly focusing on
beliefs over individual variables.

To state the result, say that \emph{$g^{*}$ influences $i\in N-G^{\prime}$}
if $F^{G^{\prime}}_{g^{*}\rightarrow i}\neq\emptyset$ and there is
a path in $F^{G^{\prime}}_{g^{*}\rightarrow i}$ from $g^{*}$ to
$i$ that contains no v-colliders (i.e., the path does not include
a sequence $k\rightarrow k'\leftarrow k''$) and does not intersect
$G$. The key is that the direction of links does not change along
the path between $g^{*}$ and $i$, and that the path does not pass
through $G$. Given that $G$ is an ancestral set, if $g^{*}$ does
not influence $i$, then $i$ is independent of $g^{*}$ according
to the distribution $p_{R_{i}}$.
\begin{prop}
\noindent\label{prop:iterated expectations}Suppose that flows exhibit
no contingent thinking and that $G$ is an ancestral set (i.e., $R(i)\subset G$
for every $i\in G$). Then,

\noindent (0) if $g^{*}$ influences no $j\in N-G^{\prime}$, then
beliefs satisfy iterated expectations;

\noindent (1) if $g^{*}$ influences only $j\in N-G^{\prime}$, then
beliefs may or may not satisfy iterated expectations; and

\noindent (2) if $g^{*}$ influences distinct $j,k\in N-G^{\prime}$,
then beliefs do not satisfy iterated expectations.
\end{prop}
Thus, when $g^{*}$ influences no variables (two variables), the iterated
expectations property holds (does not hold). When $g^{*}$ influences
exactly one variable, the property can go either way. We demonstrate
this below, and prove the result's other two cases in Appendix \ref{sec:Contingent-thinking-proofs}. 

To illustrate, suppose that $N=\{1,2,3,4\}$, $G=\{1\}$, $g^{*}=2$,
and $R(i)=\left\{ j\in N\mid j<i\right\} $. Note that $F^{G}_{2\rightarrow i}=\emptyset$
for $i=3,4$, and so $\hat{p}(x_{3},x_{4}|x_{1})=\hat{p}(x_{3}|x_{1})\hat{p}(x_{4}|x_{1}).$
That is, the DM neglects any correlation between $3$ and $4$ after
updating from $1$.

Consider case (0), where $g^{*}=2$ influences neither $3$ nor $4$.
Then, $F^{G'}_{2\rightarrow i}$ equals either $\emptyset$, $2\leftarrow1\rightarrow i$,
or $2\leftarrow1\rightarrow j\rightarrow i$ for each $i=3,4$. In
either case, $\hat{p}(x_{3},x_{4}|x_{1},x_{2})$ is constant in $x_{2}$
because $1$ blocks any information propagating from $2$ to $3$
or $4$. By Assumption 2, the flows from $1$ to $3$ and $4$ do
not change when conditioning on $2$ as well as $1$. Therefore, $\hat{p}(x_{3},x_{4}|x_{1},x_{2})=\hat{p}(x_{3},x_{4}|x_{1})$,
and iterated expectations trivially hold. 

Now consider case (2), where $g^{*}=2$ influences both $3$ and $4$.
This occurs when under $G^{\prime}$, the flow from $2$ to $i$ contains
a path that does not intersect $G$ for $i=3,4$. Therefore, $F^{G'}_{2\rightarrow i}$
equals either $2\rightarrow i$ or $2\rightarrow j\rightarrow i$
for $j\in N\setminus\left\{ 1,2,i\right\} $. The DM propagates the
flow from $2$ to both $3$ and $4$, so beliefs about both $3$ and
$4$ change with the different values of $2$, introducing correlation
between the two conditional on $1$. This rules out iterated expectations
holding for any $p$ where $3$ and $4$ are both, say, positively
correlated with $2$.

Finally, turn to case (1), where $g^{*}=2$ influences $3$ alone.
Assume that no variables flow to $4$. If $F^{G}_{1\rightarrow2}=1\rightarrow2$,
$F^{G}_{1\rightarrow3}=\emptyset$ and $F^{G^{\prime}}_{2\rightarrow3}=2\rightarrow3$,
then 
\[
\hat{p}\left(x_{3}|x_{1}\right)=p(x_{3}),\ \hat{p}\left(x_{2}|x_{1}\right)=p\left(x_{2}|x_{1}\right),\ \&\ \hat{p}(x_{3}|x_{1},x_{2})=p\left(x_{3}|x_{2}\right).
\]
Combining, iterated expectations holds only if 
\[
\sum_{x_{2}}p\left(x_{3}|x_{2}\right)p\left(x_{2}|x_{1}\right)=p\left(x_{3}\right),
\]
which fails for most $p$; for example, when both $X_{1}$ and $X_{2}$
are pairwise-correlated with $X_{3}$. However, if $F^{G}_{1\rightarrow3}=1\rightarrow3$
and all the other flows are the same as above , then 
\[
\hat{p}\left(x_{3}|x_{1}\right)=p\left(x_{3}|x_{1}\right),\ \hat{p}\left(x_{2}|x_{1}\right)=p\left(x_{2}|x_{1}\right),\ \&\ \hat{p}(x_{3}|x_{1},x_{2})=p\left(x_{3}|x_{2},x_{1}\right).
\]
Iterated expectations holds using the chain rule of probability.

\subsection{Testable implications}

The elements of the limited-propagation belief representation are
quite flexible (although in applications, we are likely to impose
strong structure). Therefore, the question may arise what restrictions
the representation imposes on the DM's subjective beliefs. We perform
an exercise demonstrating that it has testable implications.

Suppose that the true DAG $R$ over $N=\{1,...,n\}$ is the chain
$1\rightarrow\cdots\rightarrow n$. Let $G=\{1\}.$ We also assume
that the order $\succ$ is the natural order on $N$. This ensures
that the DM propagates her information to nodes that are closer to
it before nodes that are farther away. It thus ensures that information
is maximally propagated.

Under this specification, we say that limited-propagation beliefs
about node $i$ are correct if $\hat{p}(x_{i}\mid x_{1})\equiv p(x_{i}\mid x_{1})$
for every $x_{N}$ and every $p$ that is consistent with $R$.
\begin{prop}
\label{prop:testable}If limited-propagation beliefs about $i$ are
correct in the above setting, then they are correct for every $j$
between $1$ and $i$. 
\end{prop}
If limited-propagation beliefs are correct about variables that are
distant from the given variable, then they should also be correct
about closer ones. This observation demonstrates that limited-propagation
beliefs have empirical bite, under very weak restrictions on the representation's
elements. The result can be extended to apply to any DAG with a path
between $1$ and $i$ by only requiring the property to hold for probability
distributions consistent with the above DAG. That is, $p$ sets all
variables not in the path between $1$ and $i$ are independent of
all those in the path, and makes the variables on the path a Markov
chain. One could in principle explore the complete testable implications,
as \citet{ellis2025subjective} do for DAGs.

The restriction to the natural order is crucial for the result. Consider
the following counter-example. Let $n=4$, and define the order to
be $3\succ4\succ2\succ1$. The flows are as follows: $F_{1\rightarrow4}=R$;
$F_{2\rightarrow3}=2\rightarrow3$; $F_{4\rightarrow3}=3\rightarrow4$;
and add the paths from $1$ that are implied by betweenness. Under
this specification, $M(3)=\{1,2,4\}$, but we can ignore $1$ since
$2$ blocks the path from $1$ to $3$. Also, $M(2)=M(4)=\{1\}.$Therefore,
\[
\hat{p}(x_{3}\mid x_{1})=\sum_{x_{2},x_{4}}p(x_{2}\mid x_{1})p(x_{4}\mid x_{1})p(x_{3}\mid x_{2},x_{4}).
\]
This belief treats $2$ and $4$ as independent conditional on $1$.
Since this property does not hold under $R$, it follows that the
belief about $3$ is incorrect.
\begin{proof}[Proof of Proposition \ref{prop:testable}]
A chain is a special case of a tree. By Proposition \ref{prop: polytree},
limited-propagation beliefs have a single-DAG representation, where
the DAG is given by $M$ (i.e., $M(i)$ is the set of parents of $i$,
for every $i$). The proof will make use of this result.

Suppose limited-propagation beliefs are correct about $i>1$. Note
that $M(i)$ must be non-empty; otherwise, $\hat{p}(x_{i}\mid x_{1})=p(x_{i})$,
hence different from $p(x_{i}\mid x_{1})$ for almost every $p$ that
is consistent with $R$, thereby contradicting the correctness assumption.
Let $m(i)$ denote the maximal element in $M(i)$. Since $m(i)$ blocks
the $R$-path between $i$ and any $j<m(i)$, $p(x_{i}\mid x_{M(i)})=p(x_{i}\mid x_{m(i)})$.
Therefore,
\[
\hat{p}(x_{i}\mid x_{1})=\sum_{x_{m(i)}}p_{M}(x_{m(i)}\mid x_{1})p(x_{i}\mid x_{m(i)}).
\]
Since this equation holds for every $p$ that is consistent with $R$,
it follows that beliefs about $i$ are correct if and only if they
are correct about $m(i)$.

Suppose that for $k\in\left(m(i),i\right)$, beliefs about every $j\in[m(i),k)$
are correct. Note this is true for $k=m(i)+1$. By the betweenness
property of flows, the path from $m(i)$ to $k$ is in $F_{m(i)\rightarrow k}$,
so $M(k)$ is non-empty and $m(k)\geq m(i)$. By definition, $m(k)\in[m(i),k)$,
so beliefs about $m(k)$ are correct. Our earlier argument gives that
beliefs about $k$ are correct since they are correct about $m(k)$.
Induction establishes that beliefs about every $j=m(i),...,i$ are
correct. Since $i$ was selected as an arbitrary node for which beliefs
are correct, this completes the proof.
\end{proof}

\section{Applications\label{sec:Applications}}

We now turn to some applications of our framework. Up to now, $p$
was assumed to be a full-support probability distribution. The givens
include actions taken by the DM. But a DM who maximizes utility given
limited-propagation beliefs may choose to avoid some actions, which
runs against the full support assumption when the action is represented
by a node in $R$. Following \citet{spiegler2016bayesian}, we close
the gap by looking at the limit of $\varepsilon$-equilibrium. In
this setting, we focus on a single action node $a\in G$ (the set
of feasible actions is denoted $A$). We assume that $G=\{a\}\cup R(a)$,
capturing the idea that the DM's action is directly caused by her
information (namely, the variables she observes). The objective distribution
over all variables induced by the DM's strategy $\sigma=x_{G^{*}}\mapsto\sigma(x_{a}\mid x_{G^{*}})\in\Delta A$
is 
\[
p_{\sigma}=x_{N}\mapsto p\left(x_{G^{*}}\right)\sigma\left(x_{a}\mid x_{G^{*}}\right)p\left(x_{N-G}\mid x_{G}\right),
\]
where $p(x_{G^{*}})$ and $p(x_{N-G}\mid x_{G})$ are exogenous and
determined by the setting. To make our result applicable, assume these
components have full support (this can be relaxed).

The strategy $\sigma$ is a \emph{personal equilibrium} if it is the
$m\rightarrow\infty$ limit of a sequence of strategies $\sigma_{m}$
satisfying, for every $m$ and $x_{G^{*}}$, \textit{(i)} $\sigma_{m}(x_{a}\mid x_{G^{*}})\geq\frac{1}{m}$
for all $x_{a}$; and \textit{(ii)}
\[
\sigma_{m}\left(x_{a}\mid x_{G^{*}}\right)>\frac{1}{m}\implies x_{a}\in\arg\max_{x^{\prime}_{a}\in A}\sum_{x_{N-G}}\hat{p}_{\sigma_{m}}\left(x_{N-G}\mid x_{G}\right)u\left(x^{\prime}_{a},x_{N-a}\right).
\]
In other words, under the perturbed equilibrium strategy, every action
that is played with probability above $1/m$ is a subjective best-reply
with respect to the DM's limited-propagation beliefs (which are induced
by the perturbed strategy itself). Personal equilibrium is the limit
as the perturbation vanishes.

To aid exposition in this section, we abuse notation by letting the
variables themselves (in lower case form) assume the role of nodes
in the graph. This also applies to sets of parents, i.e., for a DAG
$Q$ we let $Q(a_{i})$ be random variables that are parents of the
variable $a_{i}$.

\subsection{A sequential public-good provision game\label{example public good}}

Our DM is Player $1$, the first mover in a sequential three-player
game. She makes an effort choice $a_{1}$ after observing a signal
$t$ about an underlying state of Nature $\theta$. Player $i=2,3$
acts in period $i$, choosing whether or not to veto a project after
observing $a_{i-1}$ and an almost perfect signal about $\theta$.
The project is completed if Player $1$ exerts effort and neither
of the other two veto it. The project is a public good, which provides
a gross benefit of $1$ to all players, but also inflicts state-dependent
penalties on them. When the state of Nature is $\theta=1$, a penalty
of $1$ is imposed on Players 1 and 2, but when the state is $\theta=0$,
a penalty of $1$ is imposed on Player 3. 

All variables take values in $\{0,1\}$. For player $1$, $a_{1}=1$
($a_{1}=0$) represents high (low) effort, and for $i>1$, $a_{i}=0$
($a_{i}=1$) represents (not) vetoing the project. Player $i$ does
not veto ($a_{i}=1$) only if $a_{i-1}=1$ and she believes she is
more likely to be in a ``good'' state. However, by the asymmetric
distribution of penalties, players $2$ and $3$ differ in their evaluation
of the two states: player $2$ only wants the project to go ahead
when $\theta=1$, whereas player $3$ only wants to proceed if $\theta=0$.
The DM's payoff function is $u(x)=a_{1}(\theta a_{2}a_{3}-c)$, where
$c>0$ is a known cost of effort. The data-generating process is as
follows: $p(\theta=1)=\frac{1}{2}$; $p(\theta=t^{*}\mid t=t^{*})=q\in\left(\frac{1}{2},1\right)$
for every $t^{*}$; 
\begin{align*}
 & p(a_{2}=1|a_{1}=a^{*},\theta=\theta^{*})=(1-2\varepsilon)a^{*}\theta^{*}+\varepsilon;\ and\ \\
 & p(a_{3}=1|a_{2}=a^{*},\theta=\theta^{*})=(1-2\varepsilon)(1-\theta^{*})a^{*}+\varepsilon
\end{align*}
for some $\varepsilon\in(0,\frac{q}{1+q})$, typically approximately
$0$.

For $\varepsilon$ small enough, in any sequential equilibrium with
Bayesian updating, the DM chooses $a_{1}=0$. The reason is that when
she takes action $a^{*}$, 
\[
a_{2}a_{3}=\left(\theta a^{*}\right)\cdot\left(\theta a^{*}(1-\theta)\right)=0
\]
with probability almost $1$. The intuition is simple: the DM's opponents
have diametrically opposed attitudes to the state, and so will never
coordinate on high effort, implying that the DM has no incentive to
exert effort. 

We explore whether a DM with limited-propagation beliefs may form
more optimistic beliefs that impel her to exert effort. The true DAG
is 
\[
\begin{array}{ccccc}
t & \leftarrow & \theta\\
 &  & \downarrow & \searrow\\
a_{1} & \rightarrow & a_{2} & \rightarrow & a_{3}
\end{array},
\]
a relabeling of that in Section \ref{subsec:A-multi-DAG-representation}.
The order and flows are as described there: $F_{t\rightarrow j}$
and $F_{a_{1}\rightarrow j}$ consist of the shortest paths in $R$
from $t$ or $a_{1}$ to $j$ for all $j$; $F_{a_{2}\rightarrow a_{3}}=a_{2}\rightarrow a_{3}$;
for all other $i,j$, $F_{i\rightarrow j}=\emptyset$; and $\theta\prec a_{2}\prec a_{3}$.
Her limited-propagation beliefs are thus given by Equation (\ref{eq: formula public good})
after relabeling, with $p_{\sigma}$ replacing $p$. She correctly
updates her beliefs about $\theta$ and $a_{2}$, though she misses
some correlation. When she turns to $a_{3}$, she propagates $t$
and $a_{2}$ directly to $a_{3}$, not taking into account the intermediate
link between $\theta$ and $a_{2}$. Specifically, she chooses $a_{1}$
to maximize her expected payoff given by the belief $\hat{p}(\theta,a_{2},a_{3}\mid t,a_{1})$
in Equation (\ref{eq: formula public good}). Therefore, she plays
$a_{1}=1$ when she observes $t=t^{*}$ if and only if $\hat{p}(\theta a_{2}a_{3}=1\mid t=t^{*},a_{1}=1)>c$.

Consider the strategy $\sigma_{m}(a_{1}=1|t^{*})=\frac{m-1}{m}$ for
each $t^{*}$. We verify that this converges to a personal equilibrium
where the DM plays $1$ \textit{for every signal}, as long as $\varepsilon$
small enough. To simplify notation, we drop the dependence of $p$
on $\sigma_{m}$.

To calculate $\hat{p}$, we first calculate
\[
p(a_{2}=1\mid t=1,a_{1}=1)=q(1-\varepsilon),
\]
\[
p(\theta=1\mid t=1,a_{1}=1)=q,
\]
and then calculate $p_{R_{3}}(a_{3}=1\mid a_{2}=1,t=1)$:
\[
\sum_{\theta^{*}=0,1}p(\theta=\theta^{*}\mid t=1)p(a_{3}=1\mid\theta=\theta^{*},a_{2}=1)=\left(1-q\right)+\left(2q-1\right)\varepsilon.
\]
Combining, 
\[
\hat{p}(\theta a_{2}a_{3}=1\mid t=1,a_{1}=1)=q^{2}(1-q)+q^{2}\left(3q-2\right)\varepsilon
\]
 for $\varepsilon$ small enough. Similarly, 
\[
\hat{p}(\theta a_{2}a_{3}=1\mid t=0,a_{1}=1)\approx(1-q)^{2}q.
\]
Therefore, as long as $c<q(1-q)^{2}$, the conjectured strategy maximizes
utility for limited-propagation beliefs and $\varepsilon$ sufficiently
small. Thus, playing $a_{1}=1$ at every $t$ is a personal equilibrium,
a \textit{maximal deviation} from the sequential equilibrium prediction.

We now compare this prediction to single-DAG belief formation models
(as we defined them in Section \ref{sec:general analysis}).
\begin{prop}
\label{prop:public good}For any DAG $Q$, if $\varepsilon<\frac{c(1-q)}{2-q}$,
then $p_{Q}\left(\theta a_{2}a_{3}=1\mid t=t^{*},a_{1}=a^{*}\right)<c$
for every $\left(t^{*},a^{*}\right)$ pairs.
\end{prop}
\begin{proof}
As above, the DM best replies with $a_{1}=1$ only if $p_{Q}(a_{2}a_{3}=1|a_{1}=1,t=t^{*})\geq c$.
We consider three cases for $Q(a_{3})$. To state them succinctly,
let $\bar{x}_{A}$ be the event that all variables in $A\subset\{t,a_{1},\theta,a_{2},a_{3}\}$
equal 1, except $t$ which must equal $t^{*}$ when $t\in A$.

Case 1: $a_{2}\notin Q(a_{3})$. Notice that 
\[
p_{Q}\left(\theta a_{2}a_{3}=1\mid t=t^{*},a_{1}=1\right)\leq p\left(a_{3}=1|\bar{x}_{Q(a_{3})}\right)+\varepsilon
\]
Then, since $p(a_{3}=1|\theta=\theta^{*},a_{2}=a^{*})>\varepsilon$
only if $\theta^{*}=0$ and $a^{*}=1$. If $\theta\notin Q(a_{3})$,
then 
\begin{align*}
p(a_{3}=1|\bar{x}_{Q(a_{3})}) & \leq p\left(\left(\theta,a_{2}\right)=\left(0,1\right)|\bar{x}_{Q(a_{3})}\right)+\varepsilon\\
 & =p\left(a_{2}=1|\theta=0,\bar{x}_{Q(a_{3})}\right)p\left(\theta=0|\bar{x}_{Q(a_{3})}\right)+\varepsilon\\
 & \leq2\varepsilon.
\end{align*}
The second equality comes from $p\left(a_{2}=1|\theta=0,a_{1}=a^{*}\right)=\varepsilon$
for every $a^{*}$. So $p_{Q}\left(\theta a_{2}a_{3}=1\mid t=t^{*},a_{1}=1\right)\leq3\varepsilon$. 

Case 2: $\theta\notin Q(a_{3})$. This is similar to Case 1. Since
$p(a_{3}=1|\theta=\theta^{*},a_{2}=a^{*}_{2})>\varepsilon$ only if
$\theta^{*}=0$ and $a^{*}_{2}=1$,
\begin{align*}
p(a_{3}=1|\bar{x}_{Q(a_{3})}) & \leq p(\theta=0,a_{2}=1|\bar{x}_{Q(a_{3})})+\varepsilon\\
 & =p(\theta=0|a_{2}=1,\bar{x}_{Q(a_{3})-\{a_{2}\}})p(a_{2}=1|\bar{x}_{Q(a_{3})})+\varepsilon\\
 & =p(a_{2}=1|\theta=0,\bar{x}_{Q(a_{3})-\{a_{2}\}})\frac{p(a_{2}=1|\bar{x}_{Q(a_{3})-\{a_{2}\}})}{p(\theta=0|\bar{x}_{Q(a_{3})-\{a_{2}\}})}p(a_{2}=1|\bar{x}_{Q(a_{3})})+\varepsilon\\
 & \leq\varepsilon\frac{q}{1-q}+\varepsilon
\end{align*}
The second inequality comes from noting that $p(a_{2}=1|\theta=0,\bar{x}_{Q(a_{3})-\{x\}})=\varepsilon$,
$p(a_{2}=1|\bar{x}_{Q(a_{3})})\leq1$, and $Q(a_{3})-\{a_{2}\}\subset\{t,a_{1}\}$.
Then, $p(a_{2}=1|\bar{x}_{Q(a_{3})-\{a_{2}\}})\leq q$ and $p(\theta=0|\bar{x}_{Q(a_{3})-\{a_{2}\}})\geq1-q$.
Together, $p_{Q}\left(\theta a_{2}a_{3}=1\mid t=t^{*},a_{1}=1\right)\leq\left(2+\frac{q}{1-q}\right)\varepsilon$.

Case 3:$\{\theta,a_{2}\}\subset Q(a_{3})$. Then,
\[
p_{Q}\left(\theta a_{2}a_{3}=1\mid t=t^{*},a_{1}=1\right)\leq p(a_{3}=1|a_{2}=1,\theta=1)=\varepsilon
\]
so $p_{Q}\left(\theta a_{2}a_{3}=1\mid t=t^{*},a_{1}=1\right)\leq\varepsilon$.

Therefore, the DM best replies with $a_{1}=0$ whenever $\varepsilon<\frac{(1-q)c}{2-q},\frac{c}{2}$
since $a_{1}=1$ leads to a negative payoff, and regardless of $p_{Q}\left(\cdot\right)$,
$a_{1}=0$ gives utility weakly greater than zero.
\end{proof}

Omitting links from the true DAG captures various types of coarse
reasoning. For example, omitting $\theta\rightarrow a_{2}$ captures
neglect of the correlation between these two variables (i.e., the
fact that the player who chooses $a_{2}$ is responsive to $\theta$).
Omitting $\theta\rightarrow a_{3}$ captures a DM in the spirit of
the bilateral-trade example in \citet{esponda2008behavioral}, who
tracks the correlation between $a_{2}$ and $a_{3}$ but neglects
the fact that they are both affected by $\theta$. Fully cursed beliefs
\citep{eyster2005cursed} \textit{reorient }the links that go into
$a_{2}$ and $a_{3}$ such that their origin becomes $t$ instead
of omitting links. They can be described by the DAG $Q$
\[
\begin{array}{ccccc}
\theta & \rightarrow & t & \rightarrow & a_{3}\\
 &  & \downarrow & \searrow & \uparrow\\
 &  & a_{1} &  & a_{2}
\end{array}
\]
The DM attributes the endogenous variables $a_{2}$ and $a_{3}$ to
her information rather than the latent state of Nature. 

Proposition \ref{prop:public good} applies to both coarse and cursed
DAGs. These specifications all predict that the DM always chooses
$a_{1}=0$. Thus, the prediction of limited-propagation beliefs in
the public good game is maximally different not only from sequential
equilibrium, but also from familiar models of decision-making under
misspecified subjective models.

\subsection{A social learning game\label{subsec:social-learning}}

We present a social learning game in which each player aims to infer
exogenous state variables from a partial observation of past players'
moves. In the Bayesian equilibrium of the game, each player's action
is fully revealing, and players attain their first-best outcome in
equilibrium. We consider how players with limited-propagation beliefs
play this game. The flows used capture players that have difficulty
in propagating information over $v$-colliders (as in the Monty Hall
example). These players neglect redundancy in information as in \citet{eyster2014extensive}.

Consider a sequence of players and a sequence of signals $s_{0},s_{1},\dots$where
Player $t=1,2,...$ observes signals $s_{t-1}$ and $s_{t}$ as well
as the action of her predecessor $a_{t-1}$. The $s_{t}$'s are \textit{i.i.d}
standard normal variables. Player $t\geq1$ chooses action $a_{t}\in\mathbb{R}$
to maximize
\[
u_{t}(a_{t},\vec{a},\vec{s})=-\left(a_{t}-\sum^{t}_{\tau=0}s_{\tau}\right)^{2},
\]
and so chooses the action 
\[
E[s_{t}+s_{t-1}+...+s_{0}\mid s_{t},s_{t-1},a_{t-1}].
\]
The relationships between variables can be described by the DAG 
\[
\begin{array}{ccccccccccccc}
 &  & s_{1} &  & s_{2} &  & s_{3} &  & s_{4} &  & s_{5} &  & s_{6}\\
 &  & \downarrow & \searrow & \downarrow & \searrow & \downarrow & \searrow & \downarrow & \searrow & \downarrow & \searrow & \downarrow\\
s_{0} & \rightarrow & a_{1} & \rightarrow & a_{2} & \rightarrow & a_{3} & \rightarrow & a_{4} & \rightarrow & a_{5} & \rightarrow & \dots
\end{array}
\]
where Player $t$ drops the nodes $a_{t},a_{t+1},...$ and $s_{t+1},s_{t+2},...$. 

As a benchmark, consider the (sequential) Bayesian equilibrium. Player
$t$ updates her beliefs about $s_{0},s_{1},\dots,s_{t-2}$, and then
chooses $a_{t}$ to match the expectation of their sum. Clearly, $a_{1}=s_{1}+s_{0}$.
Assume that $a_{\tau}=s_{0}+s_{1}+...+s_{\tau}$ for $\tau\leq t$.
Consider Player $t+1$. Since $a_{t}-s_{t}=\sum^{t-1}_{j=0}s_{j}$,
playing
\[
a_{t+1}=a_{t}+s_{t+1}=s_{0}+s_{1}+\dots+s_{t+1}
\]
obtains maximum utility. Therefore, it is an equilibrium strategy,
and if all players follow it, then all private information is revealed
by actions, and all players attain their first-best payoffs.

Now, suppose that Players $1,...,t-1$ adhere to the Bayesian equilibrium
strategies, while Player $t$ has limited-propagation beliefs with
flows described as follows. The flow between two variables consists
of all the paths between them in which all the edges go in the same
direction. For example, the flow from $a_{t}$ to $a_{\tau}$ for
$\tau<t$ is 
\[
a_{t}\leftarrow a_{t-1}\leftarrow...\leftarrow a_{\tau};
\]
it does not include the path $a_{t}\leftarrow s_{t-1}\rightarrow a_{t-1}\leftarrow a_{t-2}\leftarrow...\leftarrow a_{\tau}$.
Importantly, there is no flow from $s_{t-1}$ to $s_{\tau}$ for any
$\tau<t-1$: any path between the two includes a $v$-collider, such
as $s_{t-1}\rightarrow a_{t-1}\leftarrow s_{t-2}$. Thus, Player $t$
cannot propagate information over a $v$-collider, but there is no
restriction on the number of inference steps she can perform, and
she treats given and hypothetical variables symmetrically. Finally,
we order so that variables closer to the givens come first: $a_{t-2}\prec s_{t-2}\prec a_{t-3}\prec s_{t-3}\prec...$.\footnote{As long as $a_{t-2}\prec a_{t-3}\prec a_{t-4}\prec...$, the exact
placement of $s_{\tau},s_{\tau'}$ will not qualitatively affect the
analysis below.}

We now proceed to derive Player $t$'s updated belief, given $s_{t-1},a_{t-1},s_{t}$.
Let $\hat{E}\left[\cdot\right]$ be the expectation operator using
$\hat{p}\left(\cdot|a_{t-1},s_{t-1},s_{t}\right)$. Note that the
best action is 
\[
\hat{a}_{t}=\hat{E}[s_{0}]+\hat{E}[s_{1}]+...+\hat{E}[s_{t}].
\]
Clearly $\hat{E}[a_{t-1}]=a_{t-1}$, $\hat{E}[s_{t-1}]=s_{t-1}$ and
$\hat{E}[s_{t}]=s_{t}$. 

We will show by induction that for all $\tau<t-1$, 
\[
\hat{E}[a_{\tau}]=\frac{\tau+1}{t}a_{t-1}\ \&\ \hat{E}\left[s_{\tau}\right]=\frac{1}{t}a_{t-1}.
\]
Consider $\tau<t-1$. The induction hypothesis is that $\hat{E}[a_{\tau+1}]=\frac{\tau+2}{t}a_{t-1}$.
The base case is true, since $\hat{E}[a_{t-1}]=\frac{(t-1)+1}{t}a_{t-1}$. 

Notice that 
\[
R_{a_{\tau}}=\begin{array}{ccccccccc}
... & \rightarrow & a_{\tau-1} & \rightarrow & a_{\tau} & \rightarrow & a_{\tau+1} & \rightarrow & ...\\
 & \nearrow &  & \nearrow &  & \nwarrow\\
s_{\tau-2} &  & s_{\tau-1} &  &  &  & s_{\tau} &  & s_{\tau+1}
\end{array}
\]
because the other paths are not uni-directional. Since $a_{\tau+1}\prec a_{\tau}\prec s_{\tau},a_{\tau-1},s_{\tau-1},...$,
$\left\{ a_{\tau+1}\right\} $ blocks all paths in $R_{a_{\tau}}$
from $a_{\tau}$ to its predecessors, and $a_{\tau+1}$ flows to $a_{\tau}$,
\[
p_{R_{a_{\tau}}}\left(a_{\tau}|\left(s_{\tau+1},...,a_{\tau+1},...\right)_{M(a_{\tau+1})}\right)=p\left(a_{\tau}|a_{\tau+1}\right).
\]
Therefore, $\hat{E}[a_{\tau}]=\hat{E}\left[E[a_{\tau}|a_{\tau+1}]\right]$.
Because $Cov(a_{\tau},a_{\tau+1})=\tau+1$ and $\sigma^{2}(a_{\tau+1})=\tau+2$,
the formula for updating a Normal distribution gives
\[
E[a_{\tau}|a_{\tau+1}]=\frac{\tau+1}{\tau+2}a_{\tau+1}.
\]
Therefore,
\[
\hat{E}\left[a_{\tau}\right]=\frac{\tau+1}{\tau+2}\hat{E}\left[a_{\tau+1}\right]=\frac{\tau+1}{\tau+2}\frac{\tau+2}{t}a_{t-1}=\frac{\tau+1}{t}a_{t-1},
\]
completing the induction step.

Turning to $s_{\tau}$, note that 
\[
R_{s_{\tau}}=\begin{array}{ccccccccc}
... &  & a_{\tau-1} & \rightarrow & a_{\tau} & \rightarrow & a_{\tau+1} & \rightarrow & ...\\
 &  & \uparrow & \nearrow\\
s_{\tau-1} &  & s_{\tau} &  &  &  & s_{\tau+1} &  & s_{t+2}
\end{array}
\]
Since $a_{\tau}\prec s_{\tau}\prec a_{\tau-1}$, $\left\{ a_{\tau}\right\} $
blocks all paths in $R_{s_{\tau}}$ from $s_{\tau-1}$ to its predecessors,
and $a_{\tau}$ flows to $s_{\tau}$, we have 
\[
p_{R_{s_{\tau}}}\left(s_{\tau}|\left(s_{\tau+1},...,a_{\tau+1},...\right)_{M(a_{\tau+1})}\right)=p\left(s_{\tau}|a_{\tau}\right).
\]
\[
p_{R_{s_{\tau}}}\left(s_{\tau}|x_{M(a_{\tau+1})}\right)=p\left(s_{\tau}|a_{\tau}\right).
\]
Because $Cov(s_{\tau-1},a_{\tau})=1$ and $\sigma^{2}(a_{\tau})=\tau+1$,
$E\left[s_{\tau}|a_{\tau}\right]=\frac{1}{\tau+1}a_{\tau}.$ Using
the above formula,
\[
\hat{E}\left[s_{\tau-1}\right]=\hat{E}\left[\frac{1}{\tau+1}a_{\tau}\right]=\frac{1}{\tau+1}\frac{\tau+1}{t}a_{t-1}=\frac{1}{t}a_{t-1}.
\]
This holds for all $\tau<t-1$ since $\tau$ was arbitrary. We conclude
\[
\hat{a}_{t}=s_{t+1}+s_{t}+\frac{t-1}{t}a_{t}=s_{0}+s_{1}+...+s_{t-1}+\frac{2t-1}{t}s_{t}+s_{t}.
\]
That is, the limited-propagation belief double counts $s_{t}$. 

Now, consider a setting where all players have limited-propagation
beliefs. Denote the equilibrium action of player $t$ by $\hat{a}_{t}$.
By inspection, limited-propagation beliefs agree with Bayes rule for
Player 1, so $\hat{a}_{1}=a_{1}=s_{0}+s_{1}$. Then, the equilibrium
action of Player $2$ is as above, namely $\hat{a}_{2}=s_{0}+\frac{3}{2}s_{1}+s_{2}$.
For Player $3$, we can adapt the above arguments to calculate that
\[
\hat{E}\left[s_{1}\right]=\frac{6}{17}a_{2},\ \hat{E}\left[a_{1}\right]=\frac{10}{17}a_{2},\ \&\ \hat{E}\left[s_{0}\right]=\frac{1}{2}\hat{E}\left[a_{1}\right]=\frac{5}{17}a_{2}.
\]
since $cov(\hat{a}_{2},s_{1})=\frac{3}{2}$, $cov(\hat{a}_{1},s_{1})=1$,
$cov(\hat{a}_{2},\hat{a}_{1})=\frac{5}{2}$, $\sigma^{2}(\hat{a}_{1})=2$,
and $\sigma^{2}(\hat{a}_{2})=1+\frac{9}{4}+1=\frac{17}{4}.$ Conclude
\[
\hat{a}_{3}=s_{3}+s_{2}+\frac{11}{17}a_{2}=s_{3}+\frac{28}{17}s_{2}+\frac{33}{34}s_{1}+\frac{11}{17}s_{0}.
\]
We can continue recursively to solve for the actions of later players,
and get a formula of the form 
\[
a_{t}\left(s_{t},s_{t-1},a_{t-1}\right)=s_{t}+s_{t-1}+k_{t}a_{t-1}
\]
for some $k_{t}>0$. Consequently, the equilibrium weight on $s_{t-1}$
is more than $1$ for all $t$, and the redundancy neglect persists
in equilibrium. 

\section{Procedural foundation\label{sec:Procedure}}

We now present an algorithm that implements the limited-propagation
belief representation. The procedure explicitly formalizes our claim
that each term in the factorization formula for $\hat{p}(x_{N-G}\mid x_{G})$
reflects a sequence of local computations along a subset of the paths
in $R$. 

The algorithm updates beliefs over hypothetical target variables one
by one, according to the order $\succ$. It does so by asking which
variables flow to each target, what variables they directly affect
on their path to the target, which ones those directly effect, and
so on (we use the term ``affect'' probabilistically, not causally).
The algorithm computes the magnitudes of these effects using conditional
distributions, which (following the literature) we consider to be
given via a conditional probability table or computed directly from
a dataset. For the variables that flow to the target and precede it
-- whose value has been fixed in the algorithm's previous rounds
-- it restricts attention to outcomes consistent with those values.
The effects are then propagated onwards to the other variables, using
a procedure called Variable Elimination (see Chapter 9 of \citet{koller2009probabilistic}).\footnote{The Computer Science literature has examined limited versions of belief
propagation algorithms (e.g., see \citet{koller2013general} or \citet{ihler2005loopy}).
The goal in these exercises is to develop computationally tractable
approximations of exact algorithms, or to extend the method to graphical
models for which it is unable to deliver an exact computation. }

\begin{algorithm}[t]
\caption{\textsc{Belief Propagation}}
\label{alg:localupdate}\begin{algorithmic}[1]
\Procedure{LocalUpdate}{$i,\left\{F_{j\rightarrow i}\right\}_{j\neq i}, \succ,\left(x^*_j \right)_{j\prec i}, p$}
\For{each $j \in N$}
    \State $\phi_j(\cdot) \leftarrow p(x_j \mid x_{R_i (j)})$ \Comment $R_i= \cup_{k\neq i}  F_{k\rightarrow i} $
	\State $S(\phi_j) \leftarrow \{j\} \cup R_i(j)$ 
	\If { $j \in\{j\prec i: F_{j\rightarrow i} \neq \emptyset\}$   } 
   	 \State $\phi_j(\cdot) \leftarrow \phi_j(\cdot) \cdot \mathbf{1}[x_j = x^*_j]$
	\EndIf
\EndFor
\State $\Phi \leftarrow \{\phi_1, \dots, \phi_N\},\qquad S \leftarrow \{S(\phi_1),\dots,S(\phi_N)\}$
\State $q\leftarrow \textsc{SumProductEliminate}(\Phi,S,i)$
\State \Return $\dfrac{q(\cdot)}{\sum_{x'_i} q(x'_i)}$
\EndProcedure

\Procedure{SumProductEliminate}{$\Phi,S,i$}
\State $N^o \leftarrow N \setminus \{i\}$
\While{$N^o \neq \emptyset$} 
    \State Let $j^*$ be the lowest variable in $N^o$
   
    \State $\phi^*(\cdot) \leftarrow  \displaystyle\sum_{x_{j^*}} \left( \prod_{\phi \in \Phi : j^* \in S(\phi)} \phi(x_{S(\phi)}) \right)$
	\State $S(\phi^*) \leftarrow  \displaystyle \bigcup_{\phi \in \Phi : j^* \in S(\phi)} S(\phi) - \{j^*\}$
    \State $\Phi \leftarrow \{\phi^*\} \cup \Phi - \{\phi \in \Phi : j^* \in S(\phi)\}$
	\State $S \leftarrow S \cup \{ S(\phi^*)\}- \{S(\phi) : j^* \in S(\phi)\}$
    \State $N^o \leftarrow N^o \setminus \{j^*\}$
\EndWhile
\State Let $\tilde{\phi}$ be the unique remaining factor in $\Phi$  
\State \Return $\tilde{\phi}$

\EndProcedure
\end{algorithmic}
\end{algorithm}

Algorithm \ref{alg:localupdate} formally describes the process in
pseudocode. To update variable $i$, it inputs the flows to $i$,
$\{F_{j\rightarrow i}\}_{j\neq i}$, the order $\succ$, and fixed
values for the variables that precede $i$, $(x^{*}_{j})_{j\prec i}$.
Updating consists of two steps. The first uses the flows to $i$ to
determine which calculations are necessary for each variable. It considers
the paths along which information flows to $i$, and uses those paths
to determine the relevant direct causes of each other variable. For
every such variable $j$, it constructs a \emph{factor} $\phi_{j}$.
The factor $\phi_{j}$ is a function that maps each $x_{\{j\}\cup R_{i}(j)}$
to $p(x_{j}\mid x_{R_{i}(j)})$. Its \emph{scope} $S(\phi_{j})$ consists
of the set of indexes $\{j\}\cup R_{i}(j)$. If $j$ flows to $i$
and precedes it, then the factor is updated to restrict attention
to outcomes consistent with $x^{*}_{j}$, i.e., it sets $\phi_{j}(x_{S(\phi_{j})})=0$
whenever $x_{j}\neq x^{*}_{j}$.

The second step propagates the information by marginalizing out every
variable other than the target $i$, one at a time. To eliminate variable
$j$, the algorithm multiplies together all factors whose scope contains
$j$. This produces a single function that inputs $j$ and all the
other variables that appear in a factor with it. It then sums this
product over the values that $j$ can take. This results in a new
factor that no longer depends on $j$ but that retains any information
it carried about the variables to which it was linked. The process
repeats until only $i$ remains. This leaves an updated belief about
$i$, after a renormalization.

The procedure propagates information locally. When $\textsc{SumProductEliminate}$
eliminates variable $j$, only the factors whose scope contains $j$
enter the step. These variables are ``close'' to $j$ (according
to the flows). The first elimination therefore touches only the variables
that directly affect $j$ or are directly affected by it. Subsequent
eliminations may reach variables that are only indirectly related,
since the information from previously eliminated variables has been
folded into the surviving factors. Throughout, the elimination process
only ever considers variables that are related to the target via the
flows. In any given iteration, most of the variables go untouched.

Algorithm \ref{alg:localupdate} yields limited-propagation beliefs.
\begin{prop}
For flows $\left\{ F_{j\rightarrow k}\right\} _{j,k}$ and order $\succ$,
\[
\hat{p}\left(x_{N-G}|x_{G}\right)=\prod_{i\in N-G}\textsc{LocalUpdate}\left(i,\left\{ F_{j\rightarrow i}\right\} _{j\neq i},\succ,\left(x_{j}\right)_{j\prec i},p\right)\left(x_{i}\right).
\]
That is, running Algorithm \ref{alg:localupdate} outputs limited-propagation
beliefs. 
\end{prop}
\begin{proof}
Prior to running $\textsc{SumProductEliminate}$, $\Phi=\{\phi_{1},\dots,\phi_{N}\}$
satisfies
\[
\prod_{j\in N}\phi_{j}(x)=p_{R_{i}}\left(x_{1},\dots,x_{n}\text{ and }x_{j}=x^{*}_{j}\ \forall j\in M(i)\right).
\]
Summing across all indexes not equal to $i$ gives $p_{R_{i}}(x_{i},x^{*}_{M(i)})$.
Theorem 9.5 of \citet{koller2009probabilistic} establishes that the
$\textsc{SumProductEliminate}$ procedure carries out the summation,
and so returns a function that maps $x_{i}$ to $p_{R_{i}}(x_{i},x^{*}_{M(i)})$.
The normalization changes the marginal probability into a conditional
one. Multiplying over $i\in N-G$ yields Equation (\ref{eq: subjective beliefs}).
\end{proof}

The key procedural difference between Bayesian updating and limited-propagation
beliefs comes in the algorithm's first step, when the factors are
constructed. Bayesian updating is the special case where $F_{j\rightarrow i}=R$
for every $j,i$, so the factor associated with any target variable
utilizes all of its direct causes and fixes the values of all earlier
variables. The Bayesian algorithm thus utilizes the complete information
set, as well as all the full set of paths. The limited-propagation
algorithm instead uses a restricted set of paths. Its factors involve
fewer variables and bring in less of the available evidence.

This has an ambiguous effect on the total cost of computation because
it trades off two opposing forces. The dominant driver of cost is
factor size: the operation count in the elimination phase grows exponentially
in the number of variables a factor contains. On this margin, limited
propagation wins out as its factors are sparser. However, the limited-propagation
algorithm sometimes ``forgets'' variables it has already processed
and re-processes them when updating a later target. This leads to
duplicated work that the Bayesian algorithm avoids. One would expect
the factor-size effect to dominates the duplication effect because
of the exponential effect, but this is not guaranteed. 

\section{Conclusion}

This paper presents a novel approach to modeling limited ability to
reason about multiple variables. Limited-propagation beliefs express
the DM's constraints on reasoning using paths in the underlying causal
model. In particular, the flows capture intuitions about what makes
some inferential paths harder to execute than others. For instance,
they capture limited contingent thinking when those from nodes representing
hypothetical variables are sparser than those emanating from observed
ones. Building on \citet{spiegler2016bayesian}, limited-propagation
beliefs can be comfortably integrated into equilibrium analysis of
decision-theoretic, game-theoretic and market models, as illustrated
in Section \ref{sec:Applications}.

Our framework is formally static, in that time plays no special role
in the model, but it can be readily applied to simple dynamic models,
as in Section \ref{subsec:social-learning}. While a more formal treatment
of dynamic decisions is beyond the scope of this paper, we conclude
with a brief discussion of how limited-propagation beliefs relate
to other models of bounded rationality in dynamic decision problems
and games. 

Among others, \citet[ Ch. 7]{rubinstein1998modeling}, \citet{jehiel2001limited},
and \citet{Ke2019} formulate models of limited foresight in extensive-form
games or decisions. They study players who can perfectly reason about
a short horizon but resort to some criterion when forming payoff expectations
beyond the horizon. Rather than bounding how far ahead players look,
our model limits how far inferences propagate through the causal structure,
both to future variables and historical ones.

\citet{li2017obviously} studies mechanisms in which a move by a player
at some history ``obviously'' dominates another: the worst continuation
following the former move is better for the payer than the best continuation
following the latter. \citet{nagel2023measure} extend this approach
to develop a more complete ranking of the complexity of strategy-proof
mechanisms. While these approaches are based on the relation between
payoffs and contingencies, our approach is independent of the decision
problem's payoff structure, and relies entirely on the causal relations
between variables.

\citet{cohen2026sequential}'s SCE and \citet{fong2025cursed}'s $\chi$-CSE
propose dynamic extensions of Cursed Equilibrium \citep{eyster2005cursed}.
In both, players under-attribute opponents' actions to their private
information in dynamic settings where information accrues as play
unfolds. In SCE, players display ``cursedness'' with respect to
hypothetical events but not those they have already observed, whereas
in $\chi$-CSE, they do not distinguish between the two. SCE treats
near and far future unobserved nodes identically, but $\chi$-CSE
has constant cursedness within a stage, which compounds into beliefs
more cursed about the distant future than the close. Limited-propagation
beliefs instead let accuracy vary both across variables at the same
horizon and across the near-versus-far future.

\appendix

\section{Contingent thinking proofs\label{sec:Contingent-thinking-proofs}}

Throughout this appendix, denote $R^{H}_{j}=\cup_{i}F^{H}_{i\rightarrow j}$
for $H=G,G'$.

\subsection{Proof of Proposition \ref{prop:iterated expectations} when $g^{*}$
influences two or more variables}

Suppose that there are distinct $j,k\in N-G'$ so that $F^{G^{\prime}}_{g^{*}\rightarrow j'}$
contains a collider-free path that does not intersect $G$ for $j'=j,k$.
Let $i^{j'}_{1}=g^{*},\dots,i^{j'}_{n^{j'}}=j^{\prime}$ be that path
for $j'=j,k$. Because the path from $g^{*}$ to $i^{j^{\prime}}_{k^{\prime}}$
is in $F^{G^{\prime}}_{g^{*}\rightarrow i^{j^{\prime}}_{k^{\prime}}}$,
we can restrict attention to $g^{*}\tilde{R}j$ and either $g^{*}\tilde{R}k$
or $j\tilde{R}k$. Consider $p$ so that $p(x_{\{j,k,g^{*}\}})p(x_{N-\{j,k,g^{*}\}})=p(x_{N})$
for every $x_{N}$, $p(\bar{x}_{g^{*}})=\frac{1}{2}$, $p\left(x_{\{j,k,g^{*}\}}=\bar{x}_{\{j,k,g^{*}\}}\right)=1-\varepsilon$,
and $p\left(x_{j^{\prime}}=\bar{x}_{_{j^{\prime}}}|x_{k^{\prime}},x_{g^{*}}\right)=\varepsilon$
for every $x_{g^{*}}\neq\bar{x}_{g^{*}}$ and all $x_{k^{\prime}}$,
for $\varepsilon\approx0$ . For any such $p$, there are distributions
$q$ and $\left\{ r\left(\cdot|x_{g^{*}}\right)\right\} _{x_{g^{*}}}$
on $\left(X_{j},X_{k}\right)$ so that 
\[
\hat{p}(x_{j},x_{k}|x_{G})=q(x_{j})q(x_{k})
\]
for all $x_{G}$, and
\[
\hat{p}(x_{j},x_{k}|x_{G},x_{g^{*}})=\hat{p}(x_{j},x_{k}|x^{\prime}_{G},x_{g^{*}})=r(x_{j},x_{k}|x_{g^{*}})
\]
for all $x_{g^{*}},x_{G},x^{\prime}_{G}$. Then 
\[
p_{R_{j}}\left(\bar{x}_{j}|x_{M(j)}\right)=p_{R_{j}}\left(\bar{x}_{j}\right)=p\left(\bar{x}_{j}\right)=\frac{1}{2}=q\left(\bar{x}_{j}\right)
\]
 since $M(j)\subset G$ and $X_{N-G}\perp X_{G}$. If $g^{*}\tilde{R}k$,
then similar arguments show $p_{R_{k}}\left(\bar{x}_{k}\right)=\frac{1}{2}$.
If $g^{*}\tilde{R}j\tilde{R}k$, then since $X_{k}\perp X_{R_{k}(k)-\{j\}}$
and $X_{j}\perp X_{R_{k}(j)-\{k\}}$, we have
\[
p_{R_{k}}\left(\bar{x}_{k}|x_{M(k)}\right)=\sum_{x_{j}}p\left(\bar{x}_{k}|x_{j}\right)p_{R_{k}}\left(x_{j}|x_{M^{G}(k)}\right)=\varepsilon+\left(1-2\varepsilon\right)p_{R_{k}}\left(\bar{x}_{j}\right)=\frac{1}{2}.
\]
Therefore
\[
q(\bar{x}_{j})q(\bar{x}_{k})=\frac{1}{4}.
\]
Now, since $j\tilde{R}g^{*}$,
\begin{align*}
p_{R_{j}}(\bar{x}_{j}|\bar{x}_{g^{*}},x_{M(j)-\{g^{*}\}}) & =p_{R_{j}}(\bar{x}_{j}|\bar{x}_{g^{*}})\geq p(\bar{x}_{k}|\bar{x}_{g^{*}})p(\bar{x}_{j}|\bar{x}_{k},\bar{x}_{g^{*}})\geq(1-\varepsilon)^{2}.
\end{align*}
If $k\tilde{R}j$, then $p_{R_{k}}(\bar{x}_{k}|\bar{x}_{g^{*}},x_{M(k)-\{g^{*}\}})\geq(1-\varepsilon)^{2}$.
Otherwise, $k\tilde{R}j\tilde{R}g$ and so 
\[
p_{R_{k}}(\bar{x}_{k}|\bar{x}_{g^{*}},x_{M(k)-\{g^{*}\}})=p_{R_{k}}(\bar{x}_{k}|\bar{x}_{j})p_{R_{k}}(\bar{x}_{j}|\bar{x}_{g^{*}})\geq(1-\varepsilon)^{3}.
\]
Therefore,
\[
r(\bar{x}_{j},\bar{x}_{k}|\bar{x}_{g^{*}})\geq\left(1-\varepsilon\right)^{5}.
\]
Since $\hat{p}\left(\bar{x}_{g^{*}}|x_{G}\right)=\frac{1}{2}$, iterated
expectations fails when $\frac{1}{2}\left(1-\varepsilon\right)^{5}>\frac{1}{4}$.
As in Proposition \ref{prop:flows that imply bayes}, standard results
then show that Equation (\ref{eq:iterated expectations}) does not
hold for almost all $p$.

\subsection{Proof of Proposition \ref{prop:iterated expectations} when $g^{*}$
influences no variables}

Suppose that ${F^{G'}_{g^{*}\rightarrow j}}$ contains no collider-free
paths that do not intersect $G$ for all $j$. Fix $j\in N-G'$.We
show first that for all $x_{N}$, 
\[
p_{R^{G'}_{j}}\left(x_{j}|x_{M^{G'}(j)}\right)=p_{R^{G'}_{j}}\left(x_{j}|x_{M^{G'}(j)-\{g^{*}\}}\right).
\]
This claim holds trivially when $F^{G'}_{g^{*}\rightarrow j}=\emptyset$,
so assume that $F^{G'}_{g^{*}\rightarrow j}\neq\emptyset$. Let $k$
be the first node after $g^{*}$ in the path from $g^{*}$ to $j$
in $F^{G^{\prime}}_{g^{*}\rightarrow j}$. If $k=j$, then $\left(g^{*},j\right)$
or $(j,g^{*})$ is a collider-free path that does not intersect $G$,
a contradiction, so $k\neq j$. If $(g^{*},k)\in F^{G'}_{g^{*}\rightarrow j}$,
then $(g^{*},k)\in F^{G'}_{g^{*}\rightarrow k}$. Because $G$ is
ancestral, $k\notin G$. Then, $(g^{*},k)$ is a collider-free path
to $k\notin N-G^{\prime}$, a contradiction. If instead $(k,g^{*})\in F^{G'}_{g^{*}\rightarrow j}$,
then $(k,g^{*})\in F^{G'}_{g^{*}\rightarrow k}$, and by hypothesis,
$k\in G$. Therefore, every path in $R^{G^{\prime}}_{j}$ from $g^{*}$
to $j$ intersects $G$. Since $G$ is ancestral, $g^{*}$ and $j$
are d-separated by $G$, which implies that 
\[
p_{R^{G^{\prime}}_{j}}(x_{g^{*}},x_{j}|x_{G})=p_{R^{G^{\prime}}_{j}}(x_{g^{*}}|x_{G})p_{R^{G^{\prime}}_{j}}(x_{j}|x_{G})
\]
 for every $x_{N}$ \citep[see][]{pearl2009causality}. Since $M^{G^{\prime}}(j)\subset G$,
the above equation holds.

Now, we claim that $R^{G^{\prime}}_{j}=R^{G}_{j}$. Pick any edge
$(z,y)$ with $z,y\neq g^{*}$ in $R^{G^{\prime}}_{j}$. Assumption
1 immediately implies that $(z,y)\in R^{G}_{j}$ when $(z,y)\in F^{G'}_{g\rightarrow j}$
for $g\in G$. If instead $(z,y)\in F^{G'}_{g^{*}\rightarrow j}$,
then there is a path from $g^{*}$ to $j$ in $F^{G'}_{g^{*}\rightarrow j}$
containing $(z,y)$. By the above and that $G$ is ancestral, that
path starts with $g^{*}\leftarrow g$, so there is a path from $g$
to $j$ containing $(z,y)$ in $F^{G'}_{g^{*}\rightarrow j}$. By
Assumption 2, that path is also in $F^{G'}_{g\rightarrow j}$, which
equals $F^{G}_{g\rightarrow j}$ by Assumption 1. So $(z,y)\in R^{G}_{j}$.
Since $R^{G}_{j}\subset R^{G^{\prime}}_{j}$by Assumption 1 and that
$F^{G}_{g^{*}\rightarrow j}=\emptyset$, $R^{G}_{j}=R^{G^{\prime}}_{j}$.

Therefore, $\hat{p}(x_{j}|x_{G},x_{g^{*}})=\hat{p}(x_{j}|x_{G})$.
Since $j$ was arbitrary, 
\[
\hat{p}(x_{N-G'}|x_{G},x_{g^{*}})=\prod_{j\in N-G'}\hat{p}(x_{j}|x_{G})=\hat{p}(x_{N-G'}|x_{G})
\]
so iterated expectations holds.

\bibliographystyle{plainnat}
\bibliography{propagation_bib}

\end{document}